\documentclass{article}

\usepackage{fullpage}
\usepackage{amsmath,amsthm,amssymb}
\usepackage{mathabx}

\usepackage{amsfonts}
\usepackage{bbm}

\usepackage{graphicx}
\usepackage[font=footnotesize]{caption}
\usepackage[font=footnotesize]{subcaption}
\graphicspath{ {images/} }

\usepackage{lipsum} 
\usepackage{float, color}

\newtheorem{theorem}{Theorem}

\newtheorem{proposition}[theorem]{Proposition}

\newtheorem{lemma}{Lemma}
\newtheorem{remark}{Remark}

\newtheorem{assumption}{Assumption}

\def\bs{\boldsymbol}
\def\mb{\mathbf}
\def\mc{\mathcal}

\newcommand\oprocendsymbol{\hbox{$\square$}}
\newcommand\oprocend{\relax\ifmmode\else\unskip\hfill\fi\oprocendsymbol}

\makeatletter
\def\ps@pprintTitle{%
 \let\@oddhead\@empty
 \let\@evenhead\@empty
 \def\@oddfoot{}%
 \let\@evenfoot\@oddfoot}
\makeatother

\title{Robustness and Centrality in Markov-switching Networks
\thanks{This is an extended version of the conference paper~\cite{SC-VS-NEL:17g}. Beyond the results reported therein, this extended version incorporates an analysis of the discrete-time setting and derives analytical specializations of the general results for the case of switching between two graph topologies.}
\thanks{This research has been supported by ONR grant N00014-14-1-0635 and ARO grant W911NF-14-1-0431.}
}
\author{Sarah H. Cen, Vaibhav Srivastava, and Naomi Ehrich Leonard 
\thanks{S. H. Cen is with the Department of Electrical and Computer Engineering, Carnegie Mellon University, Pittsburgh, PA 15213, USA (sarahcen@andrew.cmu.edu). V. Srivastava is with the Department of Electrical and Computer Engineering, Michigan State University, East Lansing, MI 48824, USA (vaibhav@egr.msu.edu). 
N. E. Leonard is with the Department of Mechanical and Aerospace Engineering, Princeton University, Princeton, NJ 08544, USA (naomi@princeton.edu).
}}

\date{}

\begin{document}

\maketitle

\begin{abstract}
We investigate how time-varying interactions, modeled via a Markov switching graph (MSG), impact the robustness of noisy multi-agent dynamics in both continuous- and discrete-time settings. Our focus is on the steady-state performance of consensus and leader-follower tracking dynamics subject to stochastic noise. Using the framework of Markov jump linear systems (MJLS), we derive expressions for the steady-state covariance of each agent's deviation from consensus and tracking error, respectively, and use them to quantify individual and group performance as a function of the interaction graphs and the switching dynamics. We extend established notions of robustness, certainty indices, and joint centrality from static graphs to the MSG setting. To gain analytical insight, we specialize our results to systems switching between two topologies and characterize how switching influences performance. Numerical simulations further illustrate how switching topologies affects system robustness in both coordination tasks.
\end{abstract}






\section{Introduction} \label{sec:intro}

Many systems rely on the coordinated behavior of multiple agents, each performing individual tasks while interacting strategically to achieve a common objective. Such multi-agent systems appear across a wide range of domains, including biology, engineering, and economics. These systems often operate in environments characterized by stochastic disturbances and time-varying patterns of interaction. Designing effective multi-agent systems thus requires a careful understanding of how temporal variations in connectivity and the presence of noise influence both individual and collective performance.

 In this paper, we examine the influence of time-varying interaction among agents on the robustness of consensus dynamics~\cite{olfati2007consensus,ren2005survey} and the related leader-follower collective tracking problem, in the presence of stochastic noise. We model the time-varying interaction by a Markov switching graph (MSG), defined by a set of static graphs and a Markov chain that describes the transitions among these graphs.

The robustness of the consensus problem has been extensively studied for static interaction graphs among agents~\cite{young2010robustness,zelazo2015robustness,siami2016fundamental, bamieh2012coherence, xiao2007distributed,pirani2023graph}. Extending these results to time-varying graphs is challenging due to the limited analytic tractability of time-dependent dynamical systems. Consensus dynamics have been analyzed for deterministically time-varying graphs~\cite{LM:05} and for stochastically time-varying graphs~\cite{matei2013convergence, tahbaz2010consensus, touri2011ergodicity}, but these works do not address the impact of stochastic noise. In particular, tools such as contraction analysis cannot be readily applied in the presence of stochastic disturbances.

We address this gap by leveraging the structure of MSGs, a tractable subclass of stochastically time-varying graphs, to characterize the robustness properties of consensus dynamics under noise. While consensus over MSGs has been previously studied in~\cite{matei2013convergence}, the role of noise in shaping performance within such frameworks has, to the best of our knowledge, not been examined.

The noisy leader-follower collective tracking problem involves a network of agents that must track an external reference signal despite the presence of measurement and communication noise. A subset of agents, called leaders, directly observe the reference signal, while the remaining followers rely solely on information exchanged through the network. This problem has been studied primarily in the context of optimal leader selection, where the goal is to choose a leader set that optimizes collective tracking performance~\cite{patterson2010leader, lin2014algorithms, KF-NEL:15, AC-LB-RP:12}, as well as in the context of edge modifications~\cite{shrinate2024towards}. Prior work has addressed leader selection for time-varying graphs that evolve much more slowly than the consensus dynamics, or for stochastic time-varying graphs subject to random link failures~\cite{AC-LB-RP:12}.
In contrast, we study the noisy leader-follower tracking problem under MSGs, where graph switching may occur at comparable or faster timescales. Our framework does not require any time-scale separation between graph and state dynamics and accommodates a broader class of stochastic interaction models than those considered in~\cite{AC-LB-RP:12}. Similar problems have been studied in the context of sensor selection~\cite{vafaee2024real}.

Our results leverage the theory of Markov jump linear systems (MJLS)~\cite{do2012continuous, do2005discrete} to characterize robustness and leadership properties in time-varying multi-agent networks. Our main contributions are fourfold. First, we derive measures of steady-state performance for noisy consensus and leader-follower tracking dynamics under MSGs, in both continuous- and discrete-time settings. These measures quantify system and node-level deviations from ideal behavior, due to noise, as a function of the network topologies and the switching process. Second, we show how these measures can be used to extend existing notions of robustness, certainty, and centrality measures for static graphs to MSGs. 
 Third, we specialize our general results to the case of switching between two network topologies. This analysis allows us to characterize how switching topologies influences performance relative to a randomly assigned static topology. Fourth, we numerically illustrate how system performance can be different and sometimes better for the dynamic graph as compared to the static graph. 

The remainder of the paper is organized as follows.
In Section~\ref{sec:two_coord_probs}, we introduce the noisy consensus and leader-follower tracking problems under Markov switching graphs (MSGs), define the performance metrics of interest, and review relevant background on continuous-time Markov jump linear systems (MJLS), which form the foundation of our analysis. 
In Section~\ref{sec:perf_mjls}, we derive steady-state performance expressions for both coordination problems.
Section~\ref{sec:remark_err} introduces robustness, certainty, and centrality indices for MSGs. To gain additional analytic insights, we specialize to the case of switching between two graph topologies in Section~\ref{sec:two-top-msg}. Section~\ref{sec:simulations} presents numerical examples that illustrate the influence of switching dynamics on system performance.
We conclude in Section~\ref{sec:conclusion}. Extension of the results to a discrete-time setting is presented in~\ref{ap:dt}.

\subsection{Notation}

In this paper, matrices are denoted by upper-case letters and vectors by lower-case letters in bold. The $i$-th element of vector $\mathbf{x}$ is denoted by $x_i$, and the $(i,j)$-th entry of a matrix $A$ is denoted $A_{ij}$. The transpose, inverse, and pseudoinverse of a matrix $A$ are written as $A^\top$, $A^{-1}$, and $A^{+}$, respectively. The vectors $\mathbf{0}_n$ and $\mathbf{1}_n$ denote $n \times 1$ vectors of all zeros and ones, respectively. The $n \times n$ identity matrix is $\mathbb{I}_n$ and the $n \times n$ exchange matrix is $\mathbb{J}_n$. 

The matrix $A \otimes B \in \mathbb{R}^{mp \times nq}$ denotes the Kronecker product of matrices $A \in \mathbb{R}^{m \times n}$ and $B \in \mathbb{R}^{p \times q}$. The matrix $A \oplus B = (A \otimes \mathbb{I}_{n} + \mathbb{I}_{m} \otimes B ) \in \mathbb{R}^{mn \times mn}$ denotes the Kronecker sum
of matrices $A \in \mathbb{R}^{m \times m}$ and $B \in \mathbb{R}^{n \times n}$. The covariance, trace, and vectorization of a matrix A are denoted by $\text{cov}(A)$, $\text{tr}(A)$, and $\text{vec}(A)$, respectively. 
Additionally, $\text{diag}_n (A_i)$ is the block-diagonal matrix with entries $(A_1, A_2, \ldots, A_n)$.

The indicator function is denoted $\mathbbm{1}(E)$, and the probability of an event $E$ is denoted $\Pr(E)$. The $(n-1)$-dimensional probability simplex in $\mathbb{R}^n$ is denoted $\Delta_n$, and the expectation of a random variable $X$ is written as $\mathbb{E}(X)$.


\section{Two noisy distributed coordination problems under Markov switching graphs}
\label{sec:two_coord_probs}

In this section, we review the noisy consensus and noisy leader-follower reference tracking coordination problems for time-varying networks. We present our definitions of system and node errors as measures of performance. Finally, we introduce MSG and MJLS. 

\subsection{Problem Statement}

Consider a network of $m$ agents with system state $\mathbf{x}(t) \in \mathbb{R}^m$, where $x_i(t)$ is the state of agent $i$ at time $t \geq 0$. At time $t$, each agent $i$ sends and receives information from its set of neighbors ${N}_i(t)$. The resulting communication topologies are represented by the undirected, unweighted time-varying graph $\mathcal{G}(t) = (\mathcal{V},\mathcal{E}(t), Y(t))$ for the set of nodes $\mathcal{V} = \{1, \ldots, m\}$, set of edges $\mathcal{E}(t) \subseteq \mathcal{V} \times \mathcal{V}$, and adjacency matrix $Y(t) \in \mathbb{R}^{m \times m}$. Each graph node corresponds to an agent in the network, and an edge $(i,j)$ between nodes $i$ and $j$ exists at time $t$ when $j \in {N}_i(t)$. $Y_{ij}(t) = 1$ if edge $(i,j)$ exists at time $t$; otherwise, $Y_{ij}(t) = 0$. Since $\mathcal{G}(t)$ is undirected,  $(i,j)$ implies $(j,i)$, and $Y(t)$ is thus symmetric. The degree of node $i$ at time $t$ is $d_i(t) = \sum_{j = 1}^m Y_{ij}(t)$, and the degree matrix $D(t) = \text{diag}(d_1(t), d_2(t), \hdots, d_m(t))$. The Laplacian matrix of the graph at time $t$ is $L(t) = D(t) - Y(t)$.

In the \emph{noisy consensus} problem, a set of agents seeks to reach agreement over time. They resolve their differing opinions, which are given in $\mathbf{x}$, through noisy network communication. We study the network consensus problem \cite{olfati2007consensus} under stochastic noise for time-varying graphs. The associated continuous-time  dynamics are
\begin{align}
d\mathbf{x}(t) = -L(t) \mathbf{x}(t) dt + F d \mathbf{W}(t), 
\label{eq:noisy-consensus-switching}
\end{align}
where $F$ is the system noise matrix and $d \mathbf{W}(t)$ is the $m$-dimensional standard Wiener process increment. Consensus is achieved when $\mathbf{x} = \alpha \mathbf{1}_m$, where $\alpha \in \mathbb{R}$ is the agreement value.

In the \emph{noisy leader-follower reference tracking} problem, the agents seek to track an external reference signal $\theta \in \mathbb{R}$ such that $x_i$ represents agent $i$'s estimate of $\theta$.  A subset of agents assigned as the \emph{leaders} directly measure $\theta$, which is affected by noise, and the remaining agents are termed \emph{followers}. All agents exchange information via noisy communications with their neighbors. The objective of the leaders is to drive the state of every agent to $\theta$. The associated continuous-time dynamics are
\begin{align}
   \hspace{-1pt} d\mathbf{{x}}(t) = -( L(t) \mathbf{x}(t) + K(\mathbf{x}(t) - \theta \mathbf{1}_m) )dt + F  d \mathbf{W}(t), \label{eq:leader-follower-reference}
\end{align}
where the leadership matrix $K \in \mathbb{R}^{m \times m}$ is a diagonal matrix with entries $\{ \kappa_1, \kappa_2, \hdots , \kappa_m\}$. If agent $i$ is a leader, then $\kappa_i = \kappa > 0$ such that the leaders share the same gain $\kappa$. If agent $i$ is a follower, then $\kappa_i$ = 0. The cardinality of the leader set $\mathcal{K}$ is given by $|\mathcal{K}|$. Without loss of generality, we shift the origin of $\mathbf{x}$ to $\theta \mathbf{1}_m$, reducing the dynamics in \eqref{eq:leader-follower-reference} to 
\begin{align}
  d  \mathbf{{x}}(t) = -M(t) \mathbf{x}(t) dt + F d \mathbf{W}(t) , \label{eq:leader-follower-reference-2}
\end{align}
where $M(t) = L(t) + K$.



We assess network performance using the following definitions of node and system errors, which apply only to systems for which the steady-state covariance exists. Define the {\em node error} $\text{E}_i$ of node $i$  as the steady-state variance of $x_i$: 
\begin{align*}
    \text{E}_i (\Sigma_{\text{ss}}(\mb x)) = (\Sigma_{\text{ss}}(\mb x))_{ii},
\end{align*}
where $\Sigma_{\text{ss}}(\mb x)$ is the steady-state covariance of $\mathbf{x}$. Define the {\em system error} $\text{E}$  as the steady-state variance of $\mathbf{x}$: 
\begin{align*}
    \text{E}(\Sigma_{\text{ss}}(\mb x)) = \text{tr}(\Sigma_{\text{ss}}(\mb x)) = \sum_{i=1}^m \text{E}_i (\Sigma_{\text{ss}}(\mb x)).
\end{align*} 

We draw inspiration for these definitions from previous works, including \cite{fitch2013information, patterson2010leader, young2010robustness,poulakakis2012node}. For the noisy consensus problem, the error quantifies the dispersion (or distance) from consensus and is thus inversely related to its robustness of consensus, as shown in \cite{young2010robustness}. For the noisy leader-follower reference tracking problem, the error indicates the success of the leader set in driving the agents' states to the reference signal. 


\subsection{Markov Switching Graph}

We study the two coordination problems under the assumption that the network switches between a finite set of graphs $\mathbb G = \{G_1, \ldots, G_n\}$ according to a Markov chain (MC). 
The resulting time-varying graph is known as a \emph{Markov switching graph} (MSG). Under a given MSG $\mathcal{G}$, the linear dynamical systems \eqref{eq:noisy-consensus-switching}-\eqref{eq:leader-follower-reference-dt} are \emph{Markov jump linear systems} (MJLS) \cite{gupta2009networked}. Recall that the total number of possible unweighted graphs for a given finite node set is also finite. Hence, the set $\mathbb G$ is assumed to be finite without loss of generality. Let $S=\{1, \ldots, n\}$ be the graph index set. 

The graph switching behavior of the continuous-time system is governed by a continuous-time MC (CTMC). For the graph set $\mathbb G$, the CTMC is specified by the infinitesimal time-homogeneous generator matrix $\Gamma \in \mathbb{R}^{n \times n}$ with elements:
\begin{align*}
    \Gamma_{ij} = \begin{cases} 
        q_{ij}, & i \neq j ,\\
        -v_i, & i = j,
    \end{cases}  
\end{align*}
for $i, j \in S$.  Here, $v_i = \sum_{j \in S \setminus i} q_{ij}$ is the holding rate of graph $G_i$, and $q_{ij} \geq 0$ is the transition rate from $G_i$ to $G_j$. Intuitively, $q_{ij}$ 
is the rate parameter of an exponential distribution that determines the probability that the system in graph $G_i$ transitions to graph $G_j$ with time  \cite{grimmett2001probability}.  Note that all rows in $\Gamma$ sum to $0$, and $-\Gamma$ is a Laplacian matrix. 


For the CTMC with the generator matrix $\Gamma$, let $\boldsymbol{\pi}(t) \in \Delta_n$ be the probability distribution over $\mathbb{G}$ at time $t$. Specifically, $\pi_i(t)$ is the probability that the network graph is $\mathcal{G}(t) = G_i$, and $\sum_{i = 1}^n {\pi}_{i}(t) = 1$.  Furthermore, $\boldsymbol{\pi}(t) = e^{\Gamma^\top t} \boldsymbol{\pi}(0)$ \cite{grimmett2001probability}.

The remainder of our analysis makes the following two assumptions about the MSG. 
\begin{assumption}\label{as:ergodic}
The MC underlying the MSG is \emph{ergodic}.
\end{assumption}
\begin{assumption}\label{as:connected}
Every graph in the set $\mathbb G$ is unweighted, undirected, and connected. 
\end{assumption}

Under Assumption~\ref{as:ergodic}, the MC must have a unique stationary distribution $\boldsymbol{\pi}_{ss}$ such that $\lim_{t \rightarrow \infty} \boldsymbol{\pi}(t) = \boldsymbol{\pi}_{\text{ss}}$ and $\lim_{k \rightarrow \infty} \boldsymbol{\pi}(k) = \boldsymbol{\pi}_{\text{ss}}$ for the CTMC and DTMC, respectively. An ergodic CTMC guarantees that $G^\top$ has exactly one eigenvalue at $0$ with the right eigenvector $\boldsymbol{\pi}_{ss}$ (i.e., $\boldsymbol{\pi}_{ss} = e^{G^\top} \boldsymbol{\pi}_{ss}$), and all others in the left half-plane. An ergodic DTMC ensures that ${P^\top}$ has exactly one eigenvalue at $1$ with the right eigenvector $\boldsymbol{\pi}_{ss}$ (i.e., $\boldsymbol{\pi}_{ss} = P^\top \boldsymbol{\pi}_{ss}$), and all others strictly within the unit circle in the complex plane. 

Assumption~\ref{as:connected} can be relaxed under certain conditions, but for clarity of exposition, we keep it. 


\subsection{Continuous-time Markov jump linear systems} \label{sec:prelims_mjls_ct}

Consider the following continuous-time (CT) MJLS:
\begin{align}
    d\mathbf{x}(t) = -Z(t) \mathbf{x}(t)dt + F d\mathbf{W}(t),
    \label{eq:generic_mjls}
\end{align}
where $Z(t) \in \mathbb{R}^{m \times m}$ corresponds to the network graph of the MSG at time $t$ with the generator matrix $\Gamma$. Let $Z(t) =Z_i$ whenever $\mathcal G(t) = G_i$. Next, we obtain the dynamics of the mean and second moment of $\mathbf{x}(t)$ evolving according to~\eqref{eq:generic_mjls}.

For the following proposition, let the mean $\boldsymbol{\mu}(t) = \mathbb{E}[\mathbf{x}(t)]$. The contribution of graph $G_i \in \mathbb{G}$ to the mean is $\boldsymbol{\mu}^i(t) = \mathbb{E}[\mathbf{x}(t) \mathbbm{1}(\mathcal{G}(t) = G_i)]$ such that $\boldsymbol{\mu}(t) = \sum_{i = 1}^n \boldsymbol{\mu}^i(t)$. Vertically stacking the means for all graphs gives the vector 
\begin{align*}
    \boldsymbol{{\nu}}(t) = [\boldsymbol{\mu}^1(t)^\top, \boldsymbol{\mu}^2(t)^\top, \hdots , \boldsymbol{\mu}^n(t)^\top]^\top.
\end{align*}

Similarly, let the second moment $C(t) = \mathbb{E}[\mathbf{x}(t) \mathbf{x}(t)^\top]$. The contribution of graph $G_i \in \mathbb G$ to the second moment is $C^i(t) = \mathbb{E}[\mathbf{x}(t) \mathbf{x}(t)^\top \mathbbm{1}(\mathcal{G}(t) = G_i)]$ such that $C(t) = \sum_{i = 1}^n C^i(t)$. Vertically stacking the vectorized second moments for all graphs gives the vector 
\begin{align*}
    \mathbf{c}(t) = [\text{vec}(C^1(t))^\top, \text{vec}(C^2(t))^\top, \hdots , \text{vec}(C^n(t))^\top]^\top.
\end{align*}

Finally, let  $Q = F F^\top$, $\mathcal{N} =  \text{diag}_n(Z_i) - \Gamma^T\otimes \mathbb{I}_m$ and $\mathcal{M} = \text{diag}_n(Z_i \oplus Z_i) - \Gamma^\top \otimes \mathbb{I}_{m^2}$.
\begin{proposition} 
The following statements hold for the CT MJLS~\eqref{eq:generic_mjls} with an MSG satisfying Assumptions~\ref{as:ergodic} and \ref{as:connected}:
\begin{enumerate}
    \item The dynamics of the mean term $\boldsymbol{{\nu}}(t)$ are
\begin{align}
    \boldsymbol{\dot{{\nu}}}(t) = -\mathcal{N} \boldsymbol{{\nu}}(t);
    \label{eq:dyn_mean}
\end{align} 
\item The dynamics of the second moment term $\mathbf{{c}}(t)$ are
\begin{align}
    \dot{\mathbf{c}}(t) &= - \mathcal{M} \mathbf{{c}}(t) + \boldsymbol{\pi}({t}) \otimes \text{vec}({Q}).  
    \label{eq:dyn_second_moment}
\end{align}
\end{enumerate}
\label{prop_ct_cov_dyn}
\end{proposition}
\begin{proof}
    This result is standard in MJLS literature. See, for example, \cite[Proposition 3.5]{feng1992stochastic}. 
For completeness, we have included a short proof in~\ref{ap:prop_ct_cov_dyn}.
\end{proof}


\section{Performance of distributed coordination under Markov switching graphs} \label{sec:perf_mjls}

In this section, we study and interpret the performance of the two noisy coordination problems described in Section~\ref{sec:two_coord_probs} under MSGs. 
Let $\mathbf{h}_{S,1} = \text{vec}(\mathbb{I}_{m})^\top$ and
$\mathbf{h}_{S} = \mathbf{1}_n^\top \otimes \mathbf{h}_{S,1}$. Furthermore, let $\mathbf{h}_{N,1}^i = \text{vec}(\mathbb{O}^i_{m}))^\top$ and $\mathbf{h}_{N}^i = \mathbf{1}_n^\top \otimes \mathbf{h}_{N,1}^i$, where $\mathbb{O}^i_{m}$ is the $m \times m$ matrix containing all zeros except at element $(i,i)$, which takes the value $1$. Then, $\mathbf{h}_{S} \mathbf{c}$ is the trace of second moment $\mathbb{E}[\mathbf{x} \mathbf{x}^\top]$, and $\mathbf{h}_{N}^i \mathbf{c}$ is its diagonal element $(i,i)$ corresponding to node $i$. At steady state, these expressions are the contribution of the second moment to the system and node errors, 
respectively. 

\subsection{Noisy consensus under MSGs} \label{sec:robust_consensus}

Let $\mathcal N_c$ and $\mathcal M_c$ be the system matrices in~\eqref{eq:dyn_mean} and~\eqref{eq:dyn_second_moment} after specializing Proposition~\ref{prop_ct_cov_dyn} to the CT MJLS~\eqref{eq:noisy-consensus-switching}. 
\begin{lemma}
For the continuous-time noisy consensus dynamics \eqref{eq:noisy-consensus-switching} under Assumptions \ref{as:ergodic} and \ref{as:connected}, both $\mathcal{N}_c$ and $\mathcal{M}_c$ have exactly one eigenvalue at $0$ each; all other eigenvalues lie strictly in the right half-plane.
    \label{lemma_noisy_consensus_eig}
\end{lemma}
\begin{proof}
When specializing Proposition~\ref{prop_ct_cov_dyn} to the CT MJLS~\eqref{eq:noisy-consensus-switching}, $Z_i = L_i$, and \eqref{eq:generic_mjls} reduces to~\eqref{eq:noisy-consensus-switching}. Consequently, $\mathcal{N}_c = \text{diag}_n(L_i) - \Gamma^\top \otimes \mathbb{I}_m$ and $\mathcal{M}_c = \text{diag}_n(L_i\oplus L_i) - \Gamma^\top \otimes \mathbb{I}_{m^2}$. We seek to show that $\mathcal{N}_c^\top$ and $\mathcal{M}_c^\top$ are effectively the Laplacian matrices of two large graphs that contain $nm$ and $nm^2$ nodes, respectively, and each has exactly one eigenvalue at $0$. 
    
    For $\mathcal{N}_c^\top$, $\text{diag}_n(L_i)$ is the Laplacian of a disconnected graph with $n$ clusters, each of which is a connected subgraph as required. Therefore, the null space of $\text{diag}_n(L_i)$ is spanned by vectors of the form $\mathbf{a} \otimes \mathbf{1}_m$ for any $\mathbf{a} \in \mathbb R^{n}$. Because $\Gamma$ describes an ergodic CTMC, $-\Gamma \otimes \mathbb{I}_m$ is the Laplacian of a connected graph, and the null space of $-\Gamma \otimes \mathbb{I}_m$ is spanned by vectors of the form $\boldsymbol{1}_n \otimes \mathbf{b}$ for any $\mathbf{b} \in \mathbb R^{m}$.
    
    The sum of two Laplacians is also a Laplacian. Furthermore, by Lemma 3.5 in \cite{matei2013convergence}, for two Laplacian matrices $A$ and $B$, $\text{Null}(A + B) = \text{Null}(A)\; \cap\; \text{Null}(B)$. Therefore,  $\text{Null}(\mathcal{N}^\top_c)$ is the intersection of space spanned by $\mathbf{a} \otimes \mathbf{1}_m$ and $\boldsymbol{1}_n \otimes \mathbf{b}$, which is the space spanned by $ \mathbf{1}_{mn}$. Since the eigenvalues of a Laplacian are either at $0$ or lie strictly in the right half-plane~\cite{agaev2005spectra} and the nullity of $\mathcal{N}_c^\top$ is $1$, the transpose $\mathcal{N}$, which shares its eigenvalues, has exactly one eigenvalue at $0$. It can be verified that the corresponding right eigenvector is $\boldsymbol{\pi}_{\text{ss}} \otimes \mathbf{1}_m$. 
    
    It can be shown similarly that $\mathcal{M}_c$ has the unique right eigenvector $\boldsymbol{\pi}_{\text{ss}} \otimes \mathbf{1}_{m^2}$ corresponding to the eigenvalue at $0$.
\end{proof}

Due to the eigenvalues of $\mathcal{N}_c$ and $\mathcal{M}_c$ at $0$, the second moment of $\mathbf{x}$ for the CT MJLS \eqref{eq:noisy-consensus-switching} diverges. However, the diverging part corresponds to the fully correlated component of agents' states. It therefore does not contribute to the deviation from consensus, and consequently should not affect the system and node errors. We thus seek to isolate and disregard this component of the second moment attributed to network consensus. Borrowing terminology from \cite{young2014optimising}, we label the subspace of $\mathbb{R}^{m}$ spanned by $\mathbf{1}_{m}$ the \textit{consensus subspace} and its orthogonal complement the \textit{disagreement subspace}. We show that, as with the static graph \cite{young2010robustness}, the second moments of $\mathbf{x}$ in \eqref{eq:noisy-consensus-switching}  achieve bounded steady-state values when projected onto the disagreement subspace and, importantly, measure the distance from consensus as required for the calculation of the system and node errors. 

For the following propositions, let $\mathbf{x}^\perp \in \mathbb{R}^{m-1}$ represent the orthogonal projection of $\mathbf{x}$ onto the $(m - 1)$-dimensional disagreement subspace $\mathbf{1}^\perp_m$. 
We pick $V \in \mathbb{R}^{(m - 1) \times m}$ such that its rows form the orthonormal basis of $\mathbf{1}_m^\perp$ and let $\mathbb{V} = V^\top V = \mathbb{I} - \frac{1}{m} \mathbf{1}_m \mathbf{1}_m^\top$. 
Then, as in \cite{young2010robustness}, let $\mathbf{x}^\perp = V \mathbf{x}$ and $\bar{\mathbf{x}} = V^\top \mathbf{x}^\perp$ such that $\mathbf{x} = \frac{1}{m} \mathbf{1}_m  \mathbf{1}_m^\top \mathbf{x} + \bar{\mathbf{x}}$. The vector $\bar{\mathbf{x}} \in \mathbb{R}^m$ is the component of $\mathbf{x}$ orthogonal to the consensus subspace, and we refer to it as the disagreement vector. In addition, note that $\mathbf{1}_m$ is the right eigenvector of $L(t)$ and $B(k)$ associated with the eigenvalues $0$ and $1$, respectively.

From \eqref{eq:noisy-consensus-switching}, the continuous-time disagreement dynamics is
\begin{align}\label{eq:disagreement}
    d\bar{\mathbf{x}}(t) = - \mathbb{V} L(t) \bar{\mathbf{x}}(t)dt + \mathbb{V} F d\mathbf{W}(t).
\end{align}

Lastly, mimicking the notation in Section \ref{sec:prelims_mjls_ct}, let $\bar{\boldsymbol{\mu}}(t) = \mathbb{E}[\widebar{\mathbf{x}}(t)]$, $\bar{\boldsymbol{\mu}}^i(t) = \mathbb{E}[\widebar{\mathbf{x}}(t) \mathbbm{1}(\mathcal{G}(t) = G_i)]$, for each $i \in S$, and 
\begin{align*}
    \bar{\boldsymbol{\nu}}(t) = [\bar{\boldsymbol{\mu}}^1(t)^\top, \hdots , \bar{\boldsymbol{\mu}}^n(t)^\top]^\top.
\end{align*}
Let $\widebar C(t)$, $\widebar C^i(t)$ and $\widebar{\mathbf{c}}(t)$, the second moment terms of $\widebar{\mathbf{x}}$, be defined analogously. 

\begin{proposition} For the disagreement dynamics~\eqref{eq:disagreement}  under Assumptions~\ref{as:ergodic}~and~\ref{as:connected}, the following statements hold:
\begin{enumerate}
    \item the steady-state mean disagreement vector is zero, i.e., \begin{align}\label{eq:mean_consensus}
    \bar{\boldsymbol{\nu}}_{\text{ss}} = \mathbf{0}_{nm},
    \end{align}
   where $\mathbf{0}_{nm}$ is the $nm \times 1$ vector of all zeros;
    \item the steady-state second moment of the disagreement vector is 
    \begin{align}
     \bar{\mathbf{c}}_{\text{ss}} = \widebar{\mathcal{M}}_c^{+} (\boldsymbol{\pi}_{\text{ss}} \otimes \text{vec}(\mathbb{V} Q\mathbb{V})), \label{eq:second_moment_consensus}
\end{align}
where $\widebar{\mathcal{M}}_c = \text{diag}_n(\mathbb{V} L_i \mathbb{V} \oplus \mathbb{V} L_i \mathbb{V}) - \Gamma^\top \otimes \mathbb{V} \otimes \mathbb{V}$. 
\end{enumerate}
\label{prop_noisy_consensus_cov_ss}
\end{proposition}


\begin{proof} 
Recall that $\bar{\mathbf{x}} = \mathbb{V} \mathbf{x}$ and $\mathbb{V} = \mathbb{V}^2$. Therefore, $\bar{\boldsymbol{\nu}}^i = \mathbb{V} \bar{\boldsymbol{\nu}}^i$ and $\bar{\mathbf{c}}^i = (\mathbb{V} \otimes \mathbb{V}) \bar{\mathbf{c}}^i$. Specializing Proposition~\ref{prop_ct_cov_dyn}(i) to~\eqref{eq:disagreement} gives
\begin{align*}
    \dot{\bar{\boldsymbol{\nu}}}(t) = - \widebar{\mathcal{N}}_c {\bar{\boldsymbol{\nu}}}(t),
\end{align*} 
where $\widebar{\mathcal{N}}_c = \text{diag}_n(\mathbb{V} L_i \mathbb{V}) - \Gamma^\top \otimes \mathbb{V}$. First, note that $\mathbb{V} L_i \mathbb{V}$ is also a Laplacian matrix and $\mathbb{V} = \mathbb{I}_m - \frac{1}{m}\mathbf{1}_m \mathbf{1}_m^\top$. Then, by the same logic used in the proof of Lemma~\ref{lemma_noisy_consensus_eig}, $X = \text{diag}_n(\mathbb{V} L_i \mathbb{V}) - \Gamma^\top \otimes \mathbb{I}_m$ has exactly one eigenvalue at $0$ with the eigenvector $\boldsymbol{\pi}_{\text{ss}} \otimes \mathbf{1}_m$ and all others in the right half-plane. Since the CTMC is ergodic, the remaining term in $\widebar{\mathcal{N}}_c$, which is $\Gamma^\top \otimes (\frac{1}{m}\mathbf{1}_m \mathbf{1}_m^\top)$, has exactly $n - 1$ non-zero eigenvalues, all of which are positive and map to the null space of $X$. This result means that the null space of $\widebar{\mathcal{N}}_c$ is equivalent to that of $X$, i.e., $\boldsymbol{\pi}_{\text{ss}} \otimes \mathbf{1}_m$, and its non-zero eigenvalues lie in the right half-plane. Second, recall that $\bar{\mathbf{x}} = (\mathbb{I}_m - \frac{1}{m} \mathbf{1}_m  \mathbf{1}_m^\top) \mathbf{x}$. Since any component of $\mathbf{x}$ along the direction $\mathbf{1}_m$ is removed to get $\bar{\mathbf{x}}$, by its definition, ${\bar{\boldsymbol{\nu}}}$ has no component along $\boldsymbol{\pi}_{\text{ss}} \otimes \mathbf{1}_m$. Thus, the eigenvalue at 0 is inconsequential, and ${\bar{\boldsymbol{\nu}}}$ has a steady-state solution. As all other eigenvalues lie in the right half-plane, the steady-state solution is ${\bar{\boldsymbol{\nu}}}_{\text{ss}} = \mathbf{0}_{nm}$, as stated in Proposition \ref{prop_noisy_consensus_cov_ss}(i).

Similarly, specializing Proposition~\ref{prop_ct_cov_dyn}(ii) to~\eqref{eq:disagreement} gives
\begin{align*}
    \dot{\bar{\mathbf{c}}}(t) = - \widebar{\mathcal{M}}_c \bar{\mathbf{c}}(t) + \boldsymbol{\pi}(t) \otimes \text{vec}(\mathbb{V} Q \mathbb{V}).
\end{align*}
where $\widebar{\mathcal{M}}_c = \text{diag}_n(\mathbb{V} L_i \mathbb{V} \oplus \mathbb{V} L_i \mathbb{V}) - \Gamma^\top \otimes \mathbb{V} \otimes \mathbb{V}$. It can be shown analogously to $\widebar{\mathcal{N}}_c$ that $\widebar{\mathcal{M}}_c$ has exactly one eigenvalue at $0$ with eigenvector $\boldsymbol{\pi}_{\text{ss}} \otimes \mathbf{1}_{m^2}$ that is inconsequential to the evolution of $\bar{\mathbf{c}}$ and that all other eigenvalues of $\widebar{\mathcal{M}}_c$ lie in the right half-plane, meaning that there is a steady-state solution for $\bar{\mathbf{c}}(t)$. At steady state, $\dot{\bar{\mathbf{c}}}_{\text{ss}} = \mathbf{0}_{nm^2}$ and $\boldsymbol{\pi}(t) = \boldsymbol{\pi}_{\text{ss}}$. Solving for ${\bar{\mathbf{c}}}_{\text{ss}}$ gives the result that is stated in Proposition \ref{prop_noisy_consensus_cov_ss}(ii).
\end{proof}


These results imply that the steady-state second moment of the disagreement vector is equal to the steady-state covariance. 
As a result, the system and $i$-th node errors of~\eqref{eq:noisy-consensus-switching} computed on the disagreement subspace are $\mathbf{h}_S {\mathbf{\bar c}}_{\text{ss}}$  and $\mathbf{h}_N^i {\mathbf{\bar c}}_{\text{ss}}$, respectively.


\subsection{Noisy leader-follower reference tracking under MSGs} \label{sec:joint_centrality}

Let $\mathcal N_k$ and $\mathcal M_k$ be the system matrices in~\eqref{eq:dyn_mean} and~\eqref{eq:dyn_second_moment} after specializing Proposition~\ref{prop_ct_cov_dyn} to the CT MJLS \eqref{eq:leader-follower-reference-2}. 

\begin{lemma}
For the continuous-time noisy leader-follower reference tracking  dynamics~\eqref{eq:leader-follower-reference-2} under Assumptions~\ref{as:ergodic} and \ref{as:connected}, all eigenvalues of matrices $\mathcal{N}_k$ and $\mathcal{M}_k$ lie strictly in the right half-plane.
 \label{lemma_ref_tracking_eig}
\end{lemma}
\begin{proof}
When specializing Proposition~\ref{prop_ct_cov_dyn} to the CT MJLS~\eqref{eq:leader-follower-reference-2}, $Z_i = L_i + K = M_i$, and \eqref{eq:generic_mjls} reduces to~\eqref{eq:noisy-consensus-switching}. Consequently, $\mathcal{N}_k = \text{diag}_n(M_i) - \Gamma^\top \otimes \mathbb{I}_m$ and $\mathcal{M}_k = \text{diag}_n(M_i \oplus M_i) - \Gamma^\top \otimes \mathbb{I}_{m^2}$. It follows that $\mc N_k = \mc N_c + I_{n} \otimes K$. From Lemma~\ref{lemma_noisy_consensus_eig}, $\mc N_c^\top$ is a Laplacian matrix with the one-dimensional null space $\bs 1_{nm}$. Thus, $\mc N_k$ is a diagonal perturbation of $\mc N_c$ and has all eigenvalues strictly in the right half-plane as long as $|\mathcal{K}| > 0$. 

It can be shown analogously that all eigenvalues of $\mathcal{M}_k$ also lie strictly in the right half-plane.
\end{proof}

Thus, in contrast to the noisy consensus problem, the noisy leader-follower reference tracking problem does not require the separation of consensus and disagreement subspaces since the presence of leaders removes the eigenvalue at $0$ from $\mathcal{N}_k$ and $\mathcal{M}_k$ as well as that at $1$ from $\mathcal{H}_k$ and $\mathcal{A}_k$. For the following propositions, let the specialization of $\bs{\nu}(t)$ and $\mathbf{c}(t)$ in~\eqref{eq:dyn_mean} and~\eqref{eq:dyn_second_moment} to the CT MJLS \eqref{eq:leader-follower-reference-2} be $\bs{\hat \nu}(t)$ and $\mathbf{\hat c}(t)$, respectively; let the specialization of $\bs{\nu}(k)$ and $\mathbf{c}(k)$ in~\eqref{eq:dyn_mean_dt} and~\eqref{eq:dyn_second_moment_dt} to the DT MJLS \eqref{eq:leader-follower-reference-dt} be $\bs{\hat \nu}(k)$ and $\mathbf{\hat c}(k)$, respectively.  

\begin{proposition} For the leader-follower reference tracking dynamics~\eqref{eq:leader-follower-reference-2} with $|\mathcal{K}| > 0$ and under Assumptions~\ref{as:ergodic}~and~\ref{as:connected}, the following statements hold:
\begin{enumerate}
    \item the steady-state mean of the state vector is zero, i.e.,
    \begin{align}\label{eq:mean_refer}
    {\boldsymbol{\hat \nu}}_{\text{ss}} = \mathbf{0}_{nm};
    \end{align}
    \item the steady-state second moment of the state vector is
     \begin{align}
    {\mathbf{\hat c}}_{\text{ss}} = \mathcal{M}_k^{-1} (\boldsymbol{\pi}_{\text{ss}} \otimes \text{vec}(Q)),
    \label{eq:second_moment_tracking}
\end{align}
where $\mathcal{M}_k = \text{diag}_n(M_i \oplus M_i) - \Gamma^\top \otimes \mathbb{I}_{m^2}$.
\end{enumerate}
 \label{prop_ref_tracking_cov_ss}
\end{proposition}
\begin{proof}
    For the CT MJLS \eqref{eq:leader-follower-reference-2}, all eigenvalues of ${\mathcal{N}_k}$ lie strictly in the right half-plane, meaning that there is a steady-state solution for $\hat{\boldsymbol{\nu}}(t)$. Since $\dot{\hat{\boldsymbol{\nu}}}(t) = \mathbf{0}_{nm}$ at steady state, it must be true that $\hat {\boldsymbol{\nu}}_{\text{ss}} = \mathbf{0}_{nm}$ as stated in in Proposition \ref{prop_ref_tracking_cov_ss}(i). Similarly, since the eigenvalues of ${\mathcal{M}_k}$ also lie strictly in the right half-plane, there exists a steady-state solution for $\hat{\mathbf{c}}(t)$. At steady state, $\dot{\hat{\mathbf{c}}}(t) = \mathbf{0}_{nm^2}$ and $\boldsymbol{\pi}(t) = \boldsymbol{\pi}_{\text{ss}}$. Solving for $\hat{\mathbf{c}}(t)$ under these conditions yields the result for $\hat{\mathbf{c}}_{\text{ss}}$ stated in Proposition \ref{prop_ref_tracking_cov_ss}(ii).
\end{proof}


As in the previous subsection, only the steady-state second moment is needed to determine the system and $i$-th node errors of \eqref{eq:leader-follower-reference-2}, which are given by $\mathbf{h}_S {\mathbf{\hat c}}_{\text{ss}}$  and $\mathbf{h}_N^i {\mathbf{\hat c}}_{\text{ss}}$, respectively.

\medskip
Equations \eqref{eq:second_moment_consensus} 
and \eqref{eq:second_moment_tracking} 
express the relationship between performance and graph structure. Understanding these results can reveal how the graph topologies and switching behavior, which are encoded in $\mathcal{M}$ and $\mathcal{A}$, affect how well the MSG propagates information. However, unlike  $\boldsymbol{\pi}_{\text{ss}}$ and $Q$, the effects of $\mathcal{M}$ and $\mathcal{A}$ on error are difficult to interpret due to the inverse matrices in 
\eqref{eq:second_moment_consensus} 
and \eqref{eq:second_moment_tracking}. 
Intuitively, these matrices contain the most information since it is the inverse operation that, in effect, performs the mixing of graphs in $\mathbb{G}$ according to $\Gamma$ and $P$.


\section{Robustness and leadership indices for MSGs}
\label{sec:remark_err}

The robustness of a system is measured by its deviation from the desired result. The deviation from consensus is isolated by projecting the system state onto the disagreement subspace. Proposition \ref{prop_noisy_consensus_cov_ss} shows that for noisy consensus under MSGs, this disagreement state is a stochastic process that, in the limit $t \rightarrow +\infty$, has zero mean and finite covariance. The component of covariance along consensus subspace is disregarded because it is completely correlated and thus does not correspond to deviation from consensus (i.e., the covariance of  $\mathbf{x}(t)$ along the consensus subspace is spanned by $\bs 1_m \bs 1_m^\top$).

As discussed in \cite{young2010robustness}, the trace of the steady-state disagreement covariance, which we define as the system error, measures the mean squared distance of the system state from consensus and, for a fixed undirected graph, corresponds to the $\mc H_2$ norm of~\eqref{eq:disagreement}~\cite{barooah2008estimation,GFY-LS-NEL:13}. For a fixed graph, this value has been linked to graph resistances and used to assess the robustness of consensus.  In a similar spirit, we propose a new notion of \emph{robustness of consensus for Markov switching graphs} defined by $1/(\mathbf{h}_S \bar{\mathbf{c}}_{\text{ss}})$ and call it the \emph{system certainty index}. This measure quantifies the efficacy of the MSGs in achieving consensus and allows for the ordering of MSGs by performance. As in ~\cite{young2011rearranging}, it can help in designing or dynamically rearranging MSG topologies for improved performance.  
 
The noisy consensus dynamics~\eqref{eq:noisy-consensus-switching} is a continuum approximation to evidence aggregation in fixed-sample and sequential decision-making. 
Furthermore, the variance of each agent's disagreement state, given by $\bar x_i$, reflects its decision-making accuracy \cite{poulakakis2016information} and illustrates the speed-accuracy trade-off \cite{VS-NEL:13f}. Accordingly, the inverse of the steady-state variance of $\bar x_i$ is described as the \emph{node certainty index} in~\cite{poulakakis2016information}. For a fixed graph, it was shown in~\cite{poulakakis2016information} that the node certainty index is a monotone function of the node's information centrality \cite{KS-MZ:89}. In a similar spirit, we introduce a novel \emph{notion of node certainty index for MSGs} defined by $1/(\mathbf{h}_N^i \bar{\mathbf{c}}_{\text{ss}})$ for the node $i$. The node certainty index can be used to order the nodes in an MSG based on their accuracy in decision-making.


The system error in the noisy leader-follower tracking is the mean squared tracking error for all agents. Thus, the system error is a measure of the robustness of tracking in the presence of noise. 
The choice of nodes selected as leaders largely dictates the tracking performance. Consequently, the leader selection problem has received significant attention for fixed noisy consensus networks \cite{patterson2010leader,lin2014algorithms,AC-LB-RP:12, KF-NEL:15,patterson2016optimal}. For the fixed graph with $|\mathcal{K}| = 1$, \cite{KF-NEL:15} shows that the most information-central node minimizes the system error when assigned as the leader. The system error due to multiple leaders was used to define the notion of joint centrality for multiple nodes, and its relation to network resistances is explored in \cite{KF-NEL:15}. 
The system error is used to define centrality measures for fixed consensus networks in~\cite{siami2017centrality}. In a similar spirit, we define \emph{joint robustness centrality of a set of leaders $\mathcal{K}$ in MSGs} as the inverse of the system error for~\eqref{eq:leader-follower-reference-2} with leader set $\mc K$ and in the limit $\kappa \to +\infty$, i.e.,   $ \lim_{\kappa \to + \infty} 1/(\mathbf{h}_S {\mathbf{\hat c}}_{\text{ss}})$.

\section{Analysis of the Two-Topology MSG} \label{sec:two-top-msg}

In this section, we dissect the results for the simple two-topology continuous-time MSG in order to understand the mechanisms affecting the performance of a switching network. 

Let $n = 2$. Then, the generator matrix is
\begin{align}
    \Gamma = \begin{bmatrix}
        -v_1 & v_1 
        \\
        v_2 & -v_2
    \end{bmatrix} = \beta \begin{bmatrix}
        -{\pi}_{\text{ss},2} & {\pi}_{\text{ss},2} 
        \\
        {\pi}_{\text{ss},1} & -{\pi}_{\text{ss},1}
    \end{bmatrix} ,
    \label{generator_2x2}
\end{align}
where $\beta = v_1+v_2$ and $\boldsymbol{\pi}_\text{ss} = \frac{1}{\beta} [v_2 \hspace{5pt} v_1]^\top$ is the stationary distribution. Let $\alpha = \beta {\pi}_{\text{ss},1}{\pi}_{\text{ss},2}$.

Furthermore, let $\mathbb{L}_i = (\mathbb{V}L_i\mathbb{V} \oplus \mathbb{V}L_i \mathbb{V})^+$. Then, by Proposition \ref{prop_noisy_consensus_cov_ss}, $\bar{\mathbf{c}}_{\text{ss},i} = \mathbb{L}_i \text{vec}(\mathbb{V}Q\mathbb{V})$ is the vectorized disagreement covariance for the single-topology static graph $G_i$ with dynamics \eqref{eq:noisy-consensus-switching}. Lastly, we define:
\begin{align*}
    \widebar{\text{E}}_{\text{static}} &= \mathbf{h}_{S,1}({\pi}_{\text{ss},1} \bar{\mathbf{c}}_{\text{ss},1} + {\pi}_{\text{ss},2} \bar{\mathbf{c}}_{\text{ss},2}) \\
    \widebar{\Omega} &= (\mathbb{V} \otimes \mathbb{V})({\pi}_{\text{ss},2} \mathbb{L}_1 + {\pi}_{\text{ss},1} \mathbb{L}_2) \\
    \bar{\mathbf{c}}_{\text{diff}} &= (\mathbb{V} \otimes \mathbb{V}) (\bar{\mathbf{c}}_{\text{ss},2} - \bar{\mathbf{c}}_{\text{ss},1}) \\
    \overline{\mathbf{h}}_{S,\text{diff}} &= \mathbf{h_{S,1}} (\mathbb{L}_2 - \mathbb{L}_1).
\end{align*}

\begin{proposition}\label{prop:ct-2-topology}
For $n = 2$, $\beta$ sufficiently small, and under Assumptions~\ref{as:ergodic}~and~\ref{as:connected}, the system error of the noisy consensus CT MJLS \eqref{eq:noisy-consensus-switching} computed on the disagreement subspace is 
\begin{align}
\text{E} &= \widebar{\text{E}}_\text{static} - \alpha \overline{\mathbf{h}}_{S,\text{diff}} (\mathbb{I}_{m^2} + \beta \widebar{\Omega})^{-1} \bar{\mathbf{c}}_{\text{diff}}.       
\end{align}
Substituting $\mathbf{h}_{N}^i$ for $\mathbf{h}_{S}$ yields an equivalent equation for the $i$-th node error. These results follow from Proposition \ref{prop_noisy_consensus_cov_ss}. \label{two-top-prop}
\end{proposition}
\begin{proof}
  We begin by writing  
   $\widebar{\mathcal{M}}_c = R_1 - \beta R_2$, where 
   \[
   R_1 = \text{diag}_n(\mathbb{V} L_i \mathbb{V} \oplus \mathbb{V} L_i \mathbb{V}), R_2 = \begin{bmatrix}
       -{\pi}_{\text{ss},2} & {\pi}_{\text{ss},1} \\{\pi}_{\text{ss},2} & -{\pi}_{\text{ss},1}
   \end{bmatrix}
   \otimes \mathbb{V} \otimes \mathbb{V}.
   \]
Recall the system error is  $\text{E} = \mathbf{h}_S \widebar{\mathcal{M}}_c^+  T$, where $T = \boldsymbol{\pi}_{\text{ss}} \otimes \text{vec}(\mathbb{V} Q \mathbb{V})$. 

    We choose to represent $\widebar{\mathcal{M}}_c^+ = \sum_{p = 0}^\infty \beta^p U_p$. Noting that  $\widebar{\mathcal{M}}_c \widebar{\mathcal{M}}_c^+ = \Pi$, where $\Pi = \mathbb{I}_{nm^2} - \frac{1}{nm^2} \mathbf{1}_{nm^2} \mathbf{1}_{nm^2}^\top$, we expand the series and solve for the matrices $U_r$ by matching terms with the same power of $\beta$ with the result that $U_p = (R_1^+ R_2)^p R_1^+$.

Note that 
\[
R_1^+T = \begin{bmatrix}
 {\pi}_{\text{ss},1} \bar{\mathbf{c}}_{\text{ss},1} \\{\pi}_{\text{ss},2} \bar{\mathbf{c}}_{\text{ss},2}
\end{bmatrix} \text{ and } \mathbf{h}_S R_1^+T = \widebar{\text{E}}_{\text{static}}. 
\]
Some detailed calculations yield 
\begin{multline*}
    \mathbf{h}_S (R_1^+ R_2)^p = (-1)^{p-1} \mathbf{h}_S \\\begin{bmatrix}
   -{\pi}_{\text{ss},2}  (\mathbb{L}_2 - \mathbb{L}_1) \Omega^{p-1} (\mathbb{V} \otimes \mathbb{V}) & \\
   &    {\pi}_{\text{ss},1}  (\mathbb{L}_2 - \mathbb{L}_1) \Omega^{p-1} (\mathbb{V} \otimes \mathbb{V})
\end{bmatrix},
\end{multline*}
and 
\[
\mathbf{h}_S \beta (R_1^+ R_2)^p R_1^+ T = (-1)^{p-1} \alpha    \overline{\mathbf{h}}_{S,\text{diff}} \Omega^{p-1} \bar{\mathbf{c}}_{\text{diff}}.
\]
Collecting the terms and using the Neumann series, which is guaranteed to converge for small $\beta$, we have 
    \begin{align*}
        \text{E} &= \widebar{\text{E}}_{\text{static}} - \alpha \overline{\mathbf{h}}_{S,\text{diff}} (\mathbb{I}_{m^2} + \beta \widebar{\Omega})^{-1} \bar{\mathbf{c}}_\text{diff},
    \end{align*}
\end{proof}

For the next proposition, we define $\mathbb{M}_i = (M_i \oplus M_i)^+$, ${\mathbf{c}}_{\text{ss},i} = \mathbb{M}_i \text{vec}(Q)$,  and:
\begin{align*}
    {\text{E}}_{\text{static}} &= \mathbf{h}_{S,1}({\pi}_{\text{ss},1} {\mathbf{c}}_{\text{ss},1} + {\pi}_{\text{ss},2} {\mathbf{c}}_{\text{ss},2}) \\
    {\Omega} &= {\pi}_{\text{ss},2} \mathbb{M}_1 + {\pi}_{\text{ss},1} \mathbb{M}_2 \\
    {\mathbf{c}}_{\text{diff}} &= {\mathbf{c}}_{\text{ss},2} - {\mathbf{c}}_{\text{ss},1} \\
    {\mathbf{h}}_{S,\text{diff}} &= \mathbf{h_{S,1}} (\mathbb{M}_2 - \mathbb{M}_1).
\end{align*}

\begin{proposition}\label{prop:ct-2-topology-leader}
For $n = 2$, $\beta$ sufficiently small, and under Assumptions~\ref{as:ergodic}~and~\ref{as:connected},  the system error of the noisy leader-follower reference tracking CT MJLS \eqref{eq:leader-follower-reference-2} is 
\begin{align}
\text{E} &= \text{E}_{\text{static}} - \alpha \mathbf{h_{S,\text{diff}}} (\mathbb{I}_{m^2} + \beta \Omega)^+ {\mathbf{c}}_{\text{diff}}.       
\end{align}
Substituting $\mathbf{h}_{N}^i$ for $\mathbf{h}_{S}$ yields an equivalent equation for the $i$-th node error. These results follow from Proposition \ref{prop_ref_tracking_cov_ss}.
\label{two-top-prop-leader}
\end{proposition}
\begin{proof}
Repeat the proof of Proposition \ref{two-top-prop} substituting $\mathcal{M}_k$ for $\widebar{\mathcal{M}}_c$, $\mathbb{I}_m$ for $\mathbb{V}$, $\mathbb{I}_{nm^2}$ for $\Pi$, $\mathbb{M}_i$ for $\mathbb{L}_i$, $\mathbf{c}_{\text{ss},i}$ for $\bar{\mathbf{{c}}}_{\text{ss},i}$, and so on. 
\end{proof}

We now focus our discussion on the noisy consensus problem and Proposition \ref{two-top-prop}. However, our conclusions straightforwardly apply to Proposition \ref{two-top-prop-leader} as well.

The error for the two-topology MSG consists of two main components. 
The first is $\widebar{\text{E}}_\text{static}$, which represents the baseline error incurred for the simplest MSG: the non-switching network. In other words, suppose the system initially enters into graph $G_i$ with probability $\pi_{\text{ss},i}$ and remains in $G_i$ thereafter. Then, the error, which is averaged over an infinite number of instances of the MSG, would be $\widebar{\text{E}}_\text{static}$.

Consequently, the second term reflects the effect of graph switching on the error. It is helpful to note that it resembles the expression for system error: $\mathbf{h}_S \widebar{\mathcal{M}}_c^+ \text{vec}(\mathbb{V} Q \mathbb{V})$. In that case, the elements of $\text{vec}(\mathbb{V} Q \mathbb{V})$ scale with the level of noise experienced by the agents, and $\mathbf{h}_S$ is used to extract the elements that contribute to the system error; for both, greater values mean greater error. This comparison suggests that the impact of graph switching on performance is directly related to the dissimilarity of the two topologies. Not only do greater values in $\overline{\mathbf{h}}_{S,\text{diff}}$ and ${\bar{\mathbf{c}}}_{\text{diff}}$ increase the magnitude of the second term, but they do so by selectively scaling the elements of $(\mathbb{I}_{m^2} + \beta \Omega)^+$ corresponding to the largest differences.

Extending the analogy further, $(\mathbb{I}_{m^2} + \beta \Omega)$, which mirrors $\widebar{\mathcal{M}}_c$ in the system error expression, encodes the topologies and generator matrix of some network. Although this network is difficult to interpret, it is clear that $(\mathbb{I}_{m^2} + \beta \Omega)$ and its pseudoinverse are positive semi-definite. If we make the reasonable assumption that each agent is independently affected by noise such that $F = Q = \mathbb{I}$, then $\overline{\mathbf{h}}_{S,\text{diff}}^\top = \bar{\mathbf{c}}_{\text{diff}}$. By the definition of positive semi-definite matrices, 
$\overline{\mathbf{h}}_{S,\text{diff}} (\mathbb{I}_{m^2} + \beta \widebar{\Omega})^+ \bar{\mathbf{c}}_{\text{diff}} \geq 0$. Due to the negative sign before the second term, this finding means that, given two graphs and the proportion of time the system must spend in each, allowing the network to switch between them cannot hurt its performance. In short, graph switching is favorable. We hypothesize that this holds true for any reasonable system noise matrix $F$ (e.g. positive diagonal elements), and we have not encountered evidence to the contrary. Furthermore, $\alpha > 0$ scales with the sum $q_{12} + q_{21}$. As $q_{ij}$ indicates the propensity of the system to switch out of $G_i$ into $G_j$, larger values of $\alpha$ correspond to more frequent graph switching. Since $\alpha$ scales the non-negative expression $\overline{\mathbf{h}}_{S,\text{diff}} (\mathbb{I}_{m^2} + \beta \widebar{\Omega})^+ \bar{\mathbf{c}}_{\text{diff}}$, more frequent switching gives lower error.

We conclude that, for the two-topology MSG, the act of switching between graphs lowers or does not change the error. In the former case, increasing the rate of switching further improves performance; we prove this for $F = \mathbb{I}$ and speculate that it is generally true. Moreover, we find that more dissimilar topologies produce better MSGs. This result applies analogously to Proposition \ref{two-top-prop-leader}. We believe that the intuitions developed in this section apply for arbitrarily large MSGs.


\section{Numerical Illustrations} \label{sec:simulations}


In this section, we illustrate the results of the previous sections by simulating simple MSGs. For all examples, $F = \mathbb{I}$. 


First, we examine the noisy consensus dynamics~\eqref{eq:noisy-consensus-switching} under MSGs with the state space $\mathbb{G}$ comprising the line and ring graphs shown in Figs. \ref{fig:line_ring_fig}(a)-(b) and the generator matrix
\begin{align}
    \Gamma = \epsilon \begin{bmatrix}
        -r_{12} & r_{12}
        \\
        r_{21} & -r_{21}
    \end{bmatrix},
    \label{generator_2_by_2}
\end{align}
where $q_{ij} = v_i = \epsilon {r}_{ij}$. Recall that $\Gamma$'s off-diagonal elements $q_{ij}$ scale with the MSG's propensity to switch from $G_i$ to $G_j$. Accordingly, $\epsilon$ controls the network's overall rate of graph switching. We fix $\epsilon = 0.5$ and study the system and node certainty indices as functions of $r_{12}$ and $r_{21}$ in Fig. \ref{fig:line_ring_fig}(c)-(d).  Fig. \ref{fig:line_ring_fig} illustrates that, in general, the MSG benefits from spending more time in the more connected topology. In this case, the performance improves when the MSG lingers in the ring topology (low $r_{21}$) and/or transitions to the ring more frequently (high $r_{12}$). Fig. \ref{fig:line_ring_fig}(c) also suggests that if $r_{12}$ is sufficiently high compared to $r_{21}$, the performance only improves marginally for rising $r_{12}$, which reveals that lengthening the time the system spends as a ring has diminishing returns. Fig. \ref{fig:line_ring_fig}(d) demonstrates that the nodes do not benefit equally from the addition of edge $(1,5)$. In addition, the certainty curve of node 3 shows that, for this MSG, only the shortest paths between nodes affects success as node 3 does not benefit from additional longer routes to previously accessible nodes (e.g. the extra path 3-4-5-1 in the ring does not improve node 3's certainty since the shorter path 3-2-1 already exists in the line).


\begin{figure}[ht!]
\centering
\begin{subfigure}{.25\textwidth}
  \centering
  \includegraphics[width=1.1in]{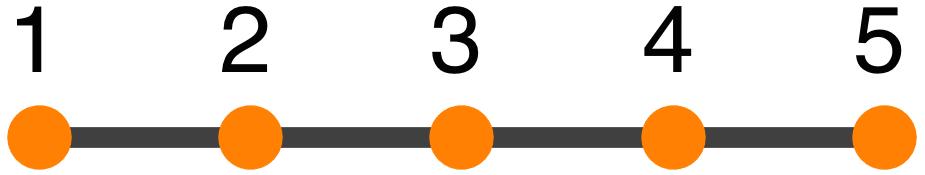}
  \caption{}
\end{subfigure}%
\begin{subfigure}{.25\textwidth}
  \centering
  \includegraphics[width=1.0in]{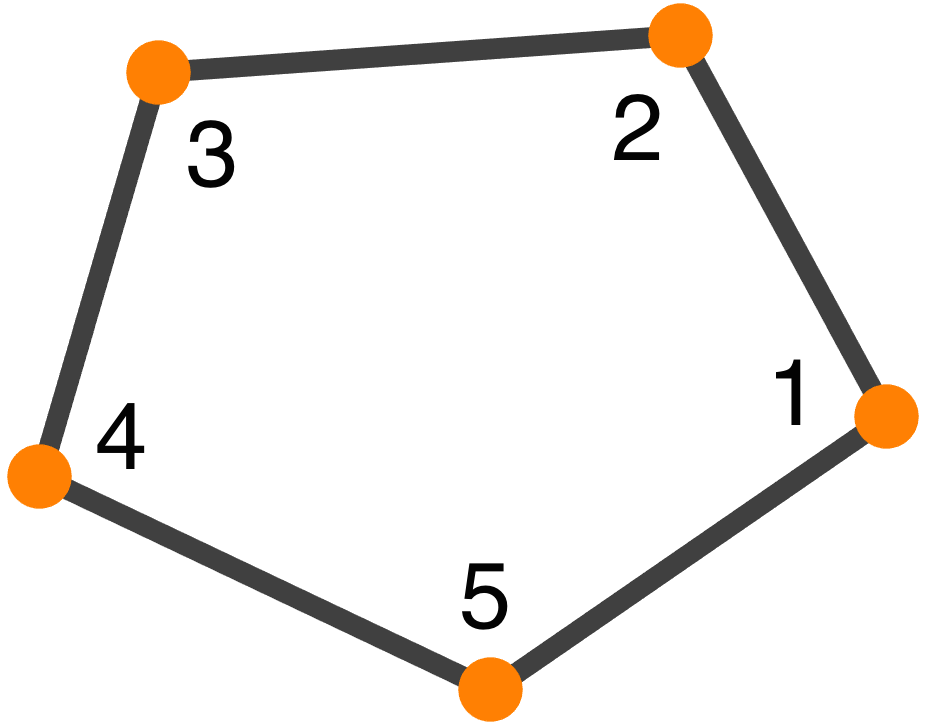}
  \caption{}
\end{subfigure}

\begin{subfigure}{.5\textwidth}
  \centering
  \hspace{-10pt}\includegraphics[height=1.45in]{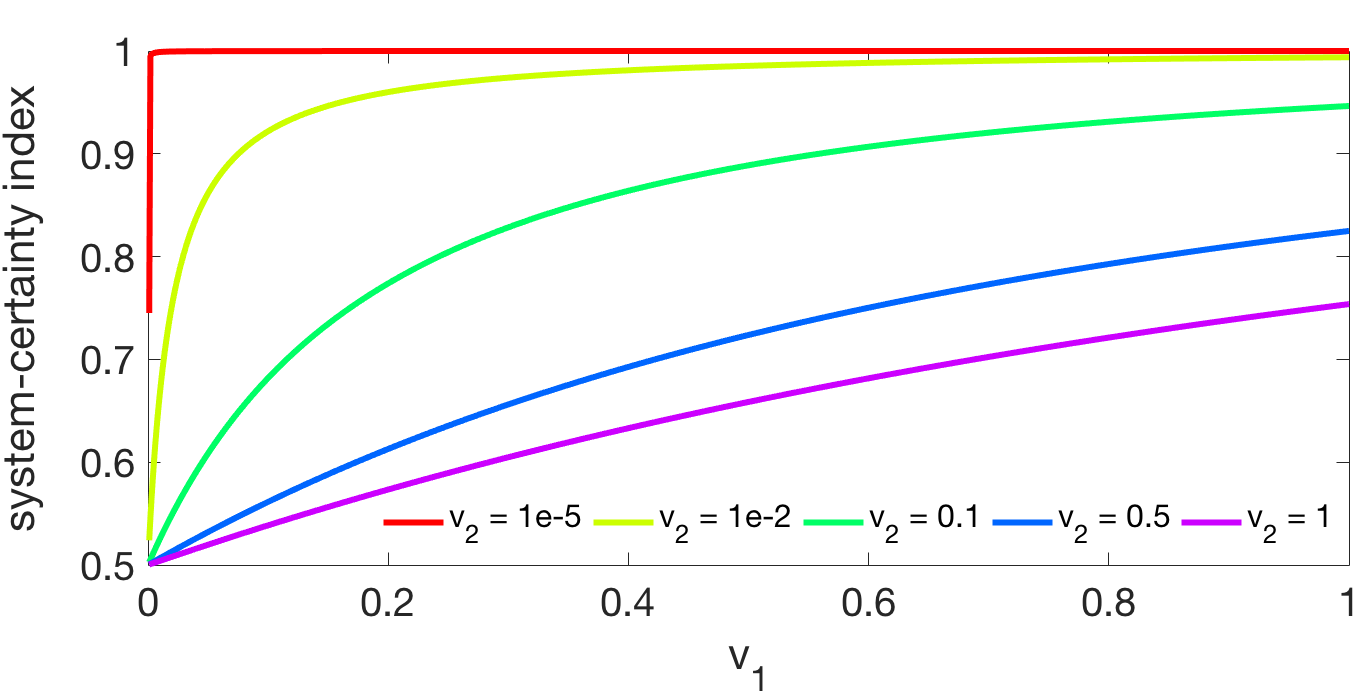}
  \caption{}
\end{subfigure}
\begin{subfigure}{.5\textwidth}
  \centering
  \hspace{-10pt}\includegraphics[height=1.45in]{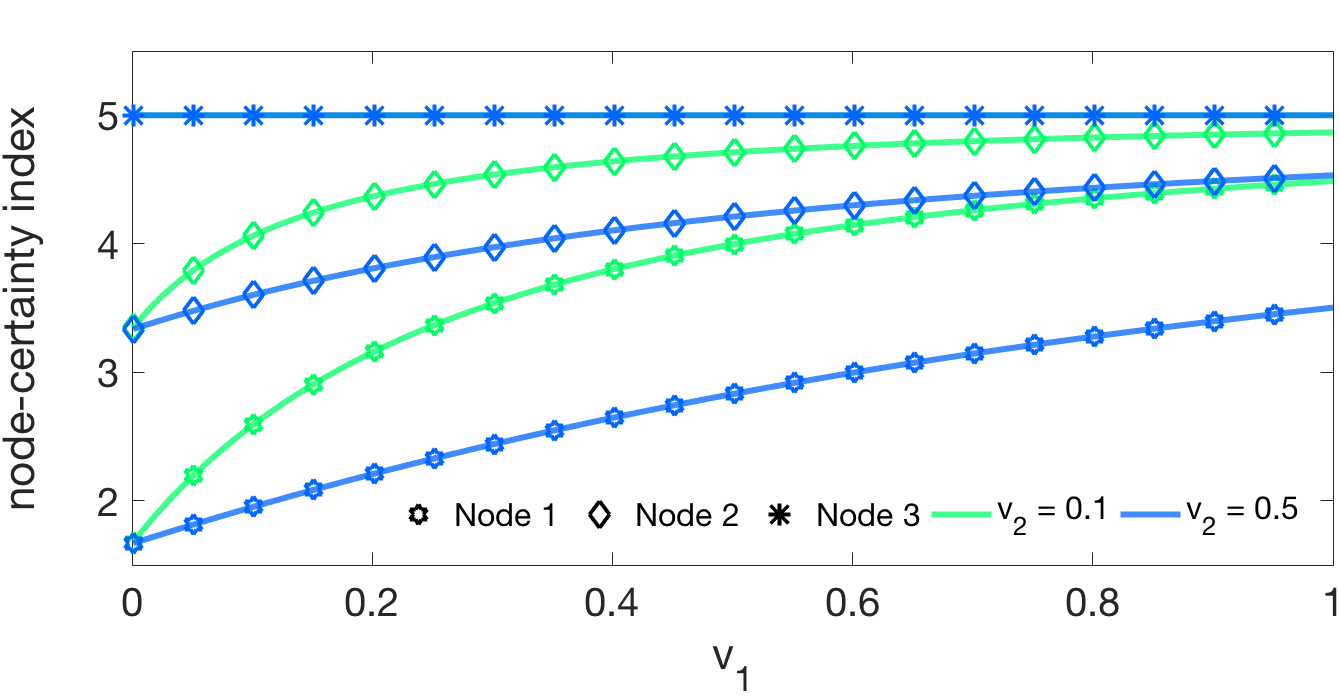}
  \caption{}
\end{subfigure}
\caption{Noisy consensus dynamics~\eqref{eq:noisy-consensus-switching} under MSG with state space comprising a line graph $G_1$ shown in panel (a) and a ring graph $G_2$ shown in panel (b), and defined by the parametrized generator matrix~\eqref{generator_2_by_2}. For fixed $\epsilon = 0.5$, panels (c) and (d) chart the system and node certainty indices, respectively, across various $r_{12}$ and $r_{21}$ values. 
}
\label{fig:line_ring_fig}
\end{figure}

Second, we investigate the noisy consensus dynamics~\eqref{eq:noisy-consensus-switching} under the MSGs with  $\mathbb G$ comprising graphs $G_1, G_2$, $G_3$, shown in Figs. \ref{fig:three_lines_fig}(a), (b), (c), respectively, and the generator matrix
\begin{align}
    \Gamma = \epsilon \begin{bmatrix}
        -r_{12} & r_{12} & 0
        \\
        {r_{2j}} & -2r_{2j} & {r_{2j}}
        \\
        0 & r_{32} & -r_{32}
    \end{bmatrix},
    \label{eq:sim_2_generator}
\end{align}
where we set $r_{12} = r_{32}$ and $r_{2j} = 0.05$. The system cannot transition directly between $G_1$ and $G_3$.
Unlike in the previous case, this MSG comprises three equally connected topologies that would individually produce identical system errors.
We illustrate system and node certainty indices as a function of $q_{12}$ and $\epsilon$ in Figs. \ref{fig:three_lines_fig}(d)-(e).
The system certainty curves in Fig. \ref{fig:three_lines_fig}(d) are concave and contain finite non-zero global maxima, which indicates that, unlike in the simpler line-graph case, the optimal MSG prefers to switch between graphs in order to balance the information flow across the nodes. Fig. \ref{fig:three_lines_fig}(e) shows three revealing facts. First, some nodes gain from increasing $q_{12}$ while others worsen, which causes the maxima in system certainty. Second, node 3 clearly does best with increasing $q_{12}$. Interestingly, node 2 also benefits from spending more time in $G_2$ despite its central placement in $G_3$, suggesting that obtaining a highly central position is less crucial than avoiding the least central ones. Finally, as $\epsilon$ increases, the certainty indices rise, meaning that, for the same stationary distribution, performance improves when switching occurs more often. 

\begin{figure}[ht!]
\centering
\vspace{1em}
\begin{subfigure}{.25\textwidth}
  \centering
  \includegraphics[width=1.1in]{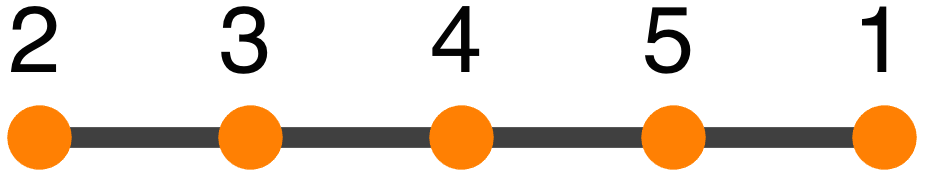}
  \caption{}
\end{subfigure}%
\begin{subfigure}{.25\textwidth}
  \centering
  \includegraphics[width=1.1in]{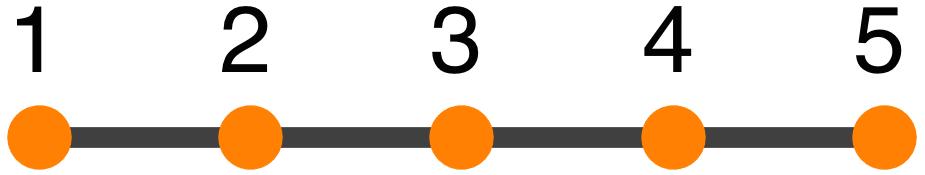}
  \caption{}
\end{subfigure}\\[0.125em]
\begin{subfigure}{.25\textwidth}
  \centering
  \includegraphics[width=1.1in]{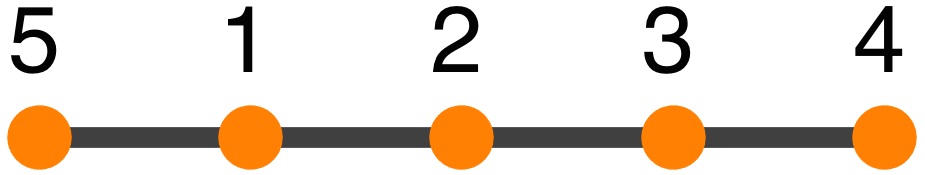}
  \caption{}
\end{subfigure}\\[0.2em]
\begin{subfigure}{.5\textwidth}
  \centering
  \hspace{-7pt}\includegraphics[height=1.45in]{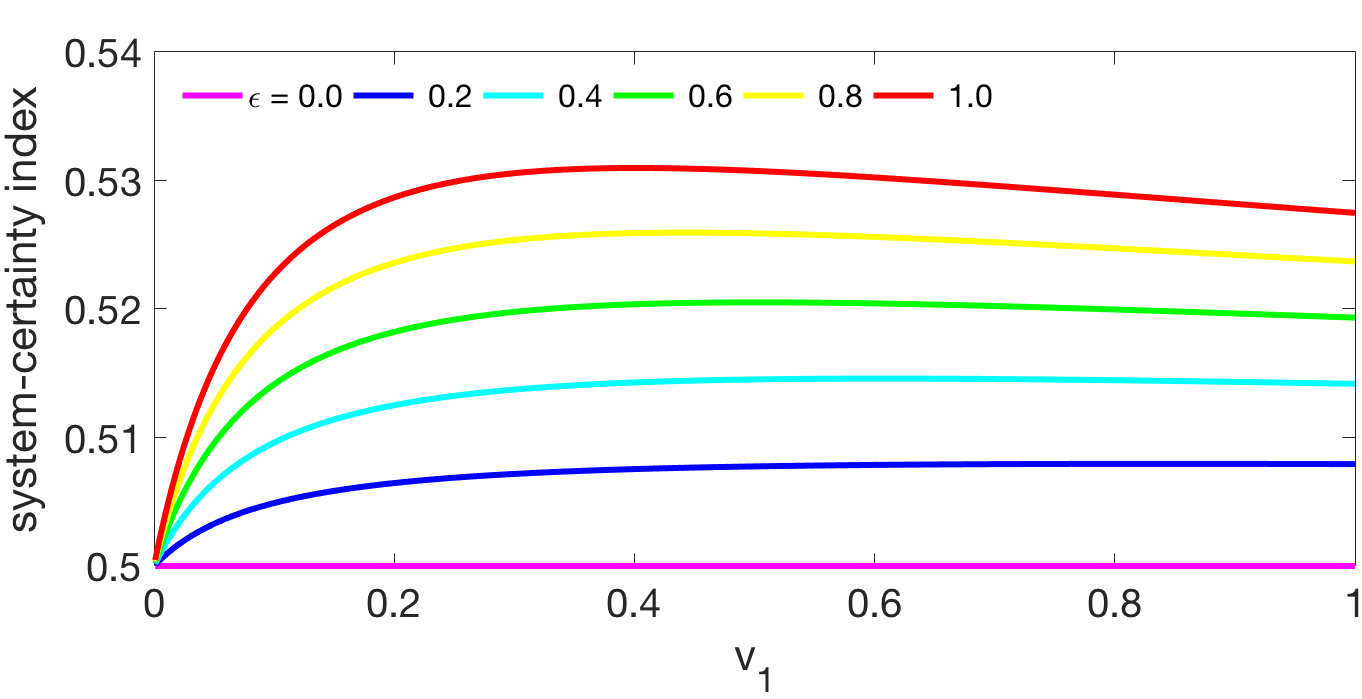}
  \caption{}
\end{subfigure}
\begin{subfigure}{.5\textwidth}
  \centering
  \hspace{-7pt}\includegraphics[height=1.45in]{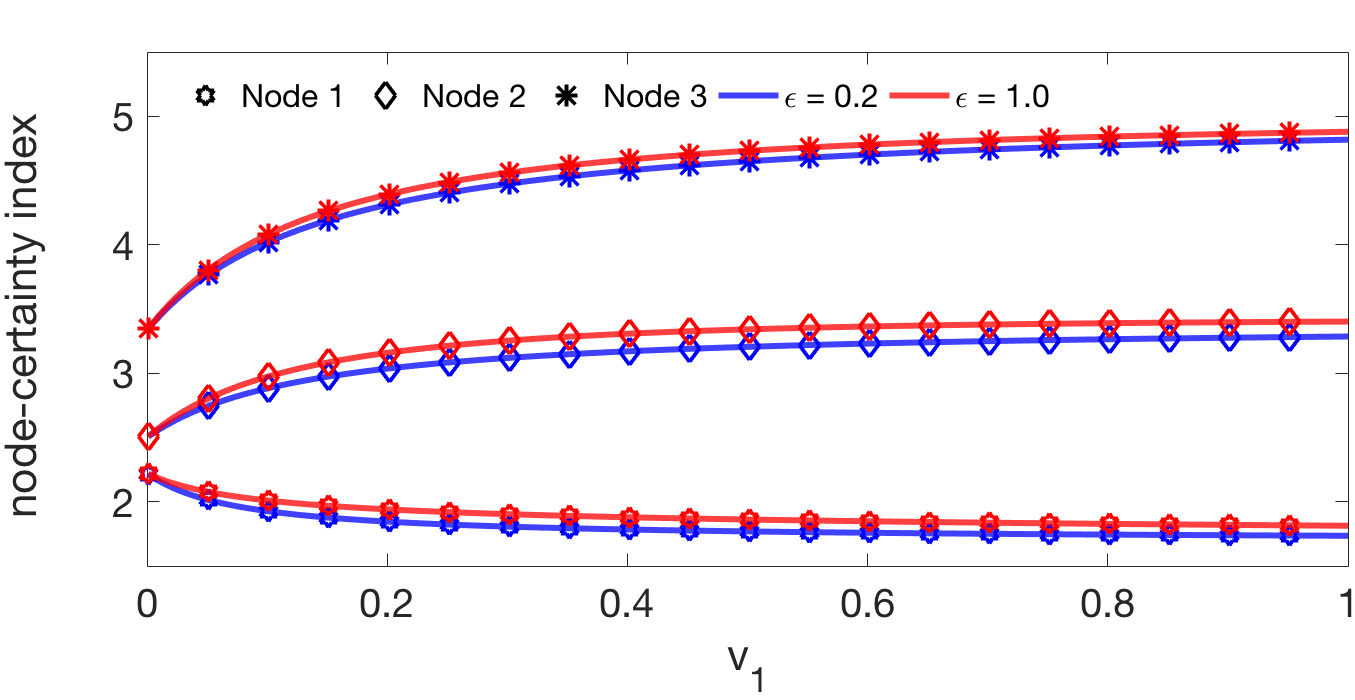}
  \caption{}
\end{subfigure}
\caption{Noisy consensus dynamics~\eqref{eq:noisy-consensus-switching} under MSG with state space comprising $G_1, G_2$ and $G_3$ shown in panels (a), (b), and (c), respectively, and defined by the parametrized generator matrix~\eqref{eq:sim_2_generator} with $v_2 = 0.1$. Panels (d) and (e) chart the system and node certainty indices, respectively, for various $v_1$ and $\epsilon$ values.
\vspace{4pt}}
\label{fig:three_lines_fig}
\end{figure}

\begin{figure}[ht!]
\centering 
\begin{subfigure}{.235\textwidth}
  \centering
  \includegraphics[width=1.0in]{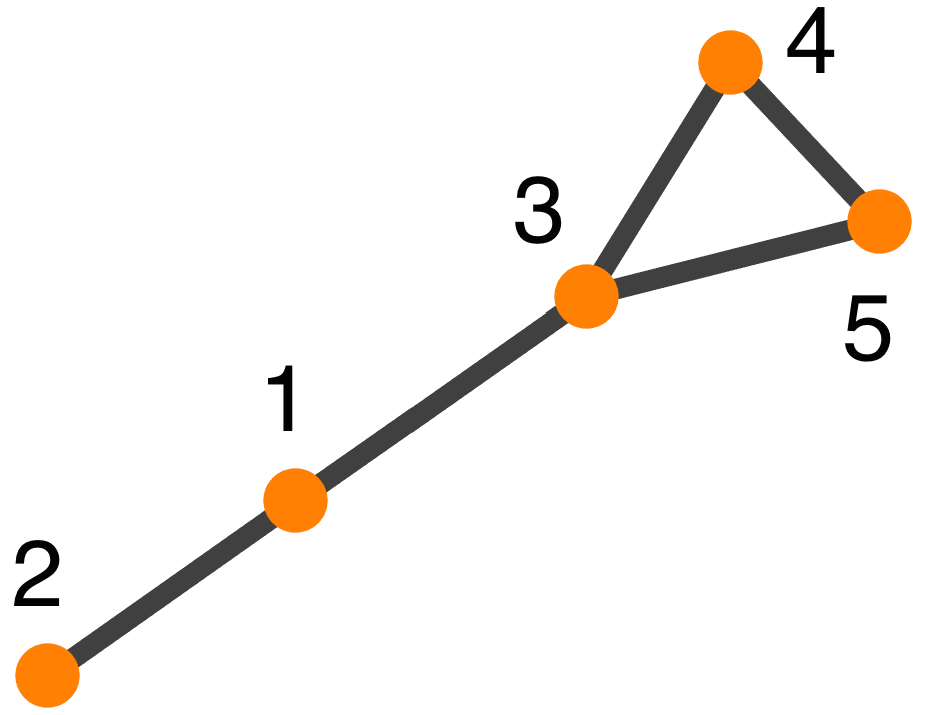}
  \caption{}
\end{subfigure}
\begin{subfigure}{.235\textwidth}
  \centering
  \includegraphics[width=1.0in]{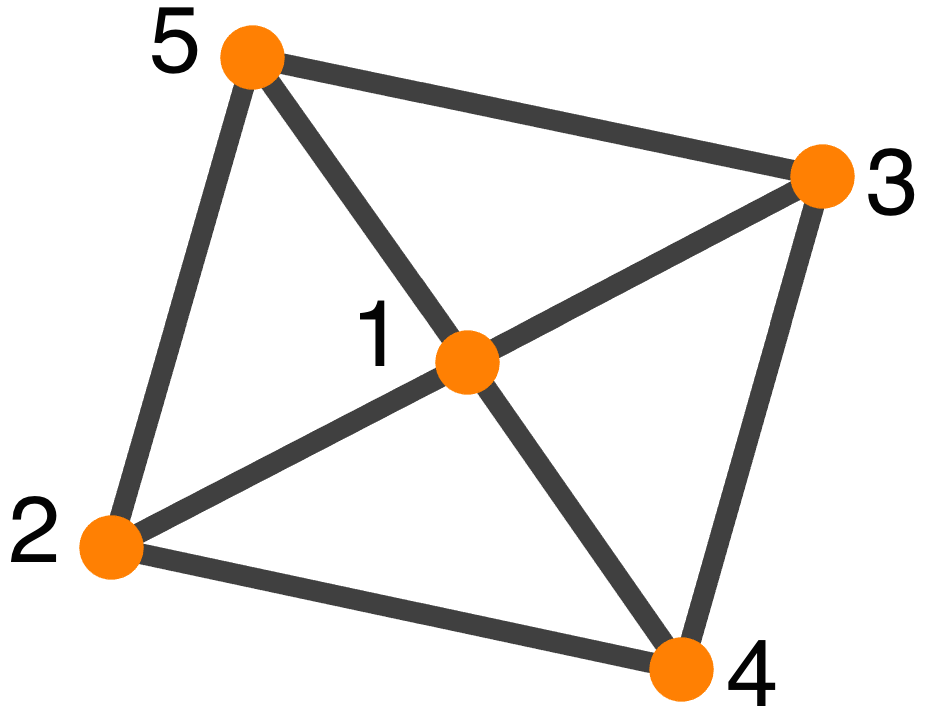}
  \caption{}
\end{subfigure}
\caption{Two topologies of a five-node MSG.} \label{fig:leader}
\end{figure}

Lastly, we study the noisy leader-follower reference tracking problem~\eqref{eq:leader-follower-reference-2} and the robustness centrality index. Consider the MSG comprising the two graphs in Fig. \ref{fig:leader} and defined by the generator matrix~\eqref{generator_2_by_2} with $q_{12} = q_{21} = 1$. 

Then, for $\kappa = \infty$, $\epsilon$ small, and $|\mathcal{K}| = 1$, the nodes listed in order of decreasing robustness centrality (i.e., leader potential) are 1-3-4-5-2. However, for $\epsilon = 1$, the nodes listed in order of decreasing node certainty are 3-1-4-5-2; interestingly, for $\epsilon = 10$, the ordering again becomes 1-3-4-5-2. Thus, the set of agents with the greatest combined node certainty does not reliably predict the optimal leader set.
This result contrasts with the static graph, for which it always holds true~\cite{KF-NEL:15}.

\section{Conclusion and Future Work} \label{sec:conclusion}
In this paper, we studied the noisy distributed consensus and noisy leader-follower reference tracking problems under Markov switching graphs (MSGs). We derived performance measures that quantify robustness of consensus, node certainty, and joint centrality in both continuous- and discrete-time settings. By specializing to the case of switching between two graph topologies, we obtained analytical insights into how switching rates and network structure influence system performance. Through numerical examples, we further illustrated how both the topology and switching dynamics affect consensus and tracking behavior. Notably, we observed that the optimal leader set for reference tracking does not necessarily correspond to the agents with the highest combined node certainty for the static graphs. This work opens several promising directions for future research. One is to better understand and characterize the structure of the matrix $\mathcal{M}$ that governs steady-state performance. Another is to develop efficient algorithms for optimal leader selection in noisy, Markov-switching networks. Additional extensions include relaxing our current assumptions to accommodate directed, weighted, or disconnected graphs.

\setcounter{section}{0}

\renewcommand \thesection {Appendix \Roman{section}}

\section{Discrete Time Consensus Networks}\label{ap:dt}

For the analogous discrete-time systems, each time step is denoted by $k \geq 0$. At some time step $k$, the network graph is given by $\mathcal{G}(k) = (\mathcal{V}, \mathcal{E}(k), Y(k))$ and parallels the continuous-time system described at the beginning of this subsection. The discrete-time dynamics for the noisy consensus problem are
\begin{align}
    \mathbf{x}(k + 1) = B(k) \mathbf{x}(k) + F \mathbf{v}(k) , \label{eq:noisy-consensus-switching-dt}
\end{align}
where $\mathbf{x}(k)$ is the system state at time step $k$, $B(k) = e^{-L(k)}$ is the averaging matrix, and $F$ is the covariance matrix for the system noise $\mathbf{v}(k) \in \mathcal{N}(\mathbf{0}_m, \mathbb{I}_m)$. 

Similarly, the discrete-time dynamics for the leader-follower reference tracking problem are
\begin{align}
    \mathbf{x}(k + 1) = A(k) \mathbf{x}(k) + F \mathbf{v}(k) , \label{eq:leader-follower-reference-dt}
\end{align}
where $A(k) = e^{-(L(k) + K)} = e^{-M(k)}$.

We first recall some results for a general discrete-time MJLS. 
Consider the following discrete-time (DT) MJLS:
\begin{align}
    \mathbf{x}(k + 1) = W(k) \mathbf{x}(k) + F \mathbf{v}(k),
    \label{eq:generic_mjls_dt}
\end{align}
where $W(k) \in \mathbb{R}^{m \times m}$ maps to the network graph of the MSG at time step $k$ with transition probability matrix $P$. Specifically, $W(k) = W_i$ whenever $\mathcal G(k) = G_i$. Next, we obtain the dynamics of the moments of $\mathbf{x}(k)$ evolving according to~\eqref{eq:generic_mjls_dt}.

The following notation aligns with that of Section \ref{sec:prelims_mjls_ct}. Let the mean $\boldsymbol{\mu}(k) = \mathbb{E}[\mathbf{x}(k)]$ and $\boldsymbol{\mu}^i(k) = \mathbb{E}[\mathbf{x}(k) \mathbbm{1}(\mathcal{G}(k) = G_i)]$ such that $\boldsymbol{\mu}(k) = \sum_{i = 1}^n \boldsymbol{\mu}^i(k)$. As for the CT MJLS, vertically stacking the means for all graphs gives the vector $\boldsymbol{{\nu}}(k)$. Furthermore, let the second moment $C(k) = \mathbb{E}[\mathbf{x}(k) \mathbf{x}(k)^\top]$ and $C^i(k) = \mathbb{E}[\mathbf{x}(k) \mathbf{x}(k)^\top \mathbbm{1}(\mathcal{G}(k) = G_i)]$ such that $C(k) = \sum_{i = 1}^n C^i(k)$. Vertically stacking the vectorized second moments for all graphs gives the vector $\mathbf{c}(k)$. Finally, let $\mathcal{H} = ({P}^T \otimes \mathbb{I}_{m}) \hspace{1pt} \text{diag}_n(W_i)$ and $\mathcal{A} = ({P}^\top \otimes \mathbb{I}_{m^2}) \hspace{1pt} \text{diag}_n(W_i \otimes W_i)$.

\begin{proposition} 
The following statements hold for the DT MJLS~\eqref{eq:generic_mjls_dt} with an MSG satisfying Assumptions~\ref{as:ergodic} and \ref{as:connected}:
\begin{enumerate}
    \item The dynamics of the mean term $\boldsymbol{{\nu}}(k)$ is
\begin{align}
    \boldsymbol{{{\nu}}}(k + 1) = \mathcal{H} \boldsymbol{{\nu}}(k);
    \label{eq:dyn_mean_dt}
\end{align} 
\item The dynamics of the second moment term $\mathbf{{c}}(k)$ is
\begin{align}
    \mathbf{{c}}(k + 1) &=  \mathcal{A} \mathbf{{c}}(k) + \boldsymbol{\pi}({k + 1}) \otimes \text{vec}({Q}).  
    \label{eq:dyn_second_moment_dt}
\end{align}
\end{enumerate}
\label{prop_dt_cov_dyn}
\end{proposition}
\begin{proof}
 This result is standard in MJLS literature. See, for example, \cite[Chapter 4]{gupta2009networked}. For completeness, we have included a short proof in~\ref{ap:prop_dt_cov_dyn}.
\end{proof}

We now establish the discrete-time equivalents of Lemma~\ref{lemma_noisy_consensus_eig} and Proposition~\ref{prop_noisy_consensus_cov_ss}. Let $\mathcal H_c$ and $\mathcal A_c$ be the system matrices in~\eqref{eq:dyn_mean_dt} and~\eqref{eq:dyn_second_moment_dt} after specializing Proposition~\ref{prop_dt_cov_dyn} to the DT MJLS~\eqref{eq:noisy-consensus-switching-dt}. In the following, the time $t$ in the notation in Section~\ref{sec:robust_consensus} is substituted with the time step $k$ for the discrete-time notation of $\bar{\boldsymbol{\mu}}(k)$, and so on. 
\begin{lemma}
For the discrete-time noisy consensus dynamics \eqref{eq:noisy-consensus-switching-dt} under Assumptions \ref{as:ergodic} and \ref{as:connected}, both $\mathcal{H}_c$ and $\mathcal{A}_c$ have exactly one eigenvalue at $1$ each; all other eigenvalues lie strictly within the unit circle of the complex plane.
    \label{lemma_noisy_consensus_eig_dt}
\end{lemma}
\begin{proof}
When specializing Proposition~\ref{prop_dt_cov_dyn} to the DT MJLS~\eqref{eq:noisy-consensus-switching-dt}, $W_i = B_i$, and \eqref{eq:generic_mjls_dt} reduces to~\eqref{eq:noisy-consensus-switching-dt}. Consequently, $\mathcal{H}_c = (P^\top \otimes \mathbb{I}_m) \text{diag}_n(B_i)$ and $\mathcal{A}_c = (P^\top \otimes \mathbb{I}_{m^2}) \text{diag}_n(B_i \otimes B_i)$. 

Let $\Lambda_1(X)$ be the eigenspace of $1$ for matrix $X$, and $|\lambda_{\setminus 1}(X)| < 1$ means that all eigenvalues of $X$ other than those at $1$ lie strictly within the unit circle of the complex plane. For an ergodic DTMC, $\Lambda_1(P) = \{\mathbf{1}_n\}$, and $|\lambda_{\setminus 1}(P)| < 1$. Therefore,  $\Lambda_1(P \otimes \mathbb{I}_m) = \{(\mathbf{1}_n \otimes \mathbf{a}) \forall \mathbf{a} \in \mathbb{R}^m \}$. In addition, since all graphs in $\mathbb{G}$ are connected, $\Lambda_1(B_i) = \{\mathbf{1}_m \}$, and $|\lambda_{\setminus 1}(B_i)| < 1$ for all $i \in S$. As a result,  $\Lambda_1(\text{diag}_n(B_i)) = \{(\mathbf{b} \otimes \mathbf{1}_m) \forall \mathbf{b} \in \mathbb{R}^n \}$. As $|\lambda_{\setminus 1}(P \otimes \mathbb{I}_m)| < 1$ and $|\lambda_{\setminus 1}(\text{diag}_n(B_i))| < 1$, it must be true that $\Lambda_1({\mathcal{H}}_c^{\hspace{1pt}\top}) = \Lambda_1(P \otimes \mathbb{I}_m) \cap \Lambda_1(\text{diag}_n(B_i)) = \{ \mathbf{1}_{mn} \}$. Its transpose ${\mathcal{H}}_c$ must also have exactly one eigenvalue at $1$. It is straightforward to verify that the associated right eigenvector is $\boldsymbol{\pi}_{\text{ss}} \otimes \mathbf{1}_m$. By the Perron-Frobenius theorem, $|\lambda_{\setminus 1}({\mathcal{H}}_c)| < 1$.

It follows similarly that $\mathcal{A}_c$ has a unique eigenvector $\boldsymbol{\pi}_{\text{ss}} \otimes \mathbf{1}_{m^2}$ corresponding to the eigenvalue at $1$ and $|\lambda_{\setminus 1}(\mathcal{A}_c)| < 1$.
\end{proof}



From \eqref{eq:noisy-consensus-switching-dt}, the discrete-time disagreement dynamics is
\begin{align}\label{eq:disagreement_dt}
    \bar{\mathbf{x}}(k + 1) = \mathbb{V} B(k) \bar{\mathbf{x}}(k) +  \mathbb{V} F \mathbf{v}(k).
\end{align}

\begin{proposition} For the disagreement dynamics~\eqref{eq:disagreement_dt}  under Assumptions~\ref{as:ergodic}~and~\ref{as:connected}, the following statements hold:
\begin{enumerate}
    \item the steady-state mean disagreement vector is zero, i.e., \begin{align}\label{eq:mean_consensus_dt}
    \bar{\boldsymbol{\nu}}_{\text{ss}} = \mathbf{0}_{nm}; 
    \end{align}
    \item the steady-state second moment of the disagreement vector is 
    \begin{align}
     \bar{\mathbf{c}}_{\text{ss}} =  (\mathbb{I}_{nm^2} - \widebar{\mathcal{A}}_c)^{+} (\boldsymbol{\pi}_{\text{ss}} \otimes \text{vec}(\mathbb{V} Q \mathbb{V})), \label{eq:second_moment_consensus_dt}
\end{align}
where $\widebar{\mathcal{A}}_c = (P^\top \otimes \mathbb{V} \otimes \mathbb{V}) \text{diag}_n (\mathbb{V} B_i\mathbb{V} \otimes \mathbb{V} B_i\mathbb{V})$.
\end{enumerate}
\label{prop_noisy_consensus_cov_ss_dt}
\end{proposition}


\begin{proof} 
For the following proof parallels, refer to the proofs of Proposition~\ref{prop_noisy_consensus_cov_ss} and  Lemma~\ref{lemma_noisy_consensus_eig_dt} for details when needed. 

Specializing Proposition~\ref{prop_dt_cov_dyn}(i) to~\eqref{eq:disagreement} gives 
\begin{align*}
    {\bar{\boldsymbol{\nu}}}(k + 1) = - \widebar{\mathcal{H}}_c {\bar{\boldsymbol{\nu}}}(k),
\end{align*}
where $\widebar{\mathcal{H}}_c = (P^\top \otimes \mathbb{I}_{m}) \text{diag}_n(\mathbb{V} B_i)$.  Since $V$ is the projector onto the disagreement subspace, $\mathbb{V} B_i$ removes the eigenvalue of $B_i$ at $1$ and replaces it with an eigenvalue at $0$. As a result,  all eigenvalues of $\text{diag}_n(\mathbb{V} B_i)$ lie strictly within the unit circle of the complex plane. From the proof of Lemma~\ref{lemma_noisy_consensus_eig_dt}, we conclude that all eigenvalues of $\widebar{\mathcal{H}}_c$ also lie strictly within the unit circle as well, meaning that the solution at steady state is ${\bar{\boldsymbol{\nu}}}_{\text{ss}} = \mathbf{0}_{nm}$, as stated in Proposition \ref{prop_noisy_consensus_cov_ss_dt}(i).

Similarly, specializing Proposition~\ref{prop_dt_cov_dyn}(ii) to~\eqref{eq:disagreement_dt} gives
\begin{align*}
    \bar{{\mathbf{c}}}(k + 1) = \bar{\mathcal{A}}_c \bar{\mathbf{c}}(k) + \boldsymbol{\pi}(k+1) \otimes \text{vec}(\mathbb{V} Q \mathbb{V}).
\end{align*}
where $\bar{\mathcal{A}}_c = (P^\top \otimes \mathbb{I}_{m^2}) \text{diag}_n(\mathbb{V} B_i \oplus \mathbb{V} B_i)$. It can be shown analogously to $\widebar{\mathcal{H}}_c$ that all eigenvalues of $\bar{\mathcal{A}}_c$ fall inside the unit circle, and a steady solution for $\bar{\mathbf{c}}(k)$ exists. At steady state, ${\bar{\mathbf{c}}}_{\text{ss}} = {\bar{\mathbf{c}}}(k + 1) = {\bar{\mathbf{c}}}(k)$ and $\boldsymbol{\pi}(k) = \boldsymbol{\pi}_{\text{ss}}$. Solving for ${\bar{\mathbf{c}}}_{\text{ss}}$ gives the result that is stated in Proposition \ref{prop_noisy_consensus_cov_ss_dt}(ii).
\end{proof}

We now extend Lemma~\ref{lemma_ref_tracking_eig} and Proposition~\ref{prop_ref_tracking_cov_ss} to the discrete-time setting. Let $\mathcal H_k$ and $\mathcal A_k$ be the system matrices in~\eqref{eq:dyn_mean_dt} and~\eqref{eq:dyn_second_moment_dt} after specializing Proposition~\ref{prop_dt_cov_dyn} to the DT MJLS~\eqref{eq:leader-follower-reference-dt}.

\begin{lemma}
For the discrete-time noisy leader-follower reference tracking  dynamics~\eqref{eq:leader-follower-reference-dt} under Assumptions~\ref{as:ergodic} and \ref{as:connected}, all eigenvalues of matrices $\mathcal{N}_k$ and $\mathcal{M}_k$ lie strictly within the unit circle of the complex plane.
 \label{lemma_ref_tracking_eig_dt}
\end{lemma}
\begin{proof}
When specializing Proposition~\ref{prop_dt_cov_dyn} to the DT MJLS~\eqref{eq:leader-follower-reference-dt}, $W_i = e^{-(L_i + K)} = A_i$, and \eqref{eq:generic_mjls_dt} reduces to~\eqref{eq:noisy-consensus-switching-dt}. Consequently, $\mathcal{H}_k = (P^\top \otimes \mathbb{I}_m)  \text{diag}_n(A_i)$ and $\mathcal{A}_k = (P^\top \otimes \mathbb{I}_{m^2}) \text{diag}_n(A_i \otimes A_i)$. For $|\mathcal{K}| > 0$, the eigenvalues of $L_i + K$ lie strictly in the right half-plane, meaning that those of $A_i$ lie strictly within the unit circle. Using the same logic as in the proof of Lemma~\ref{lemma_noisy_consensus_eig_dt}, all eigenvalues of $\mathcal{H}_k$ also lie strictly within the unit circle. 

It can be shown analogously that all eigenvalues of $\mathcal{A}_k$ lie strictly within the unit circle. 
\end{proof}

\begin{proposition} For the leader-follower reference tracking dynamics~\eqref{eq:leader-follower-reference-dt} with $|\mathcal{K}| > 0$ and under Assumptions~\ref{as:ergodic}~and~\ref{as:connected}, the following statements hold:
\begin{enumerate}
    \item the steady-state mean of the state vector is zero, i.e.,
    \begin{align}\label{eq:mean_refer_dt}
    {\boldsymbol{\hat \nu}}_{\text{ss}} = \mathbf{0}_{nm};
    \end{align}
    \item the steady-state second moment of the state vector is
     \begin{align}
    {\mathbf{\hat c}}_{\text{ss}} = (\mathbb{I}_{nm^2} - \mathcal{A}_k)^{-1} (\boldsymbol{\pi}_{\text{ss}} \otimes \text{vec}(Q)),
    \label{eq:second_moment_tracking_dt}
\end{align}
where $\mathcal{A}_k = (P^\top \otimes \mathbb{I}_{m^2}) \text{diag}_n(A_i \otimes A_i)$.
\end{enumerate}
 \label{prop_ref_tracking_cov_ss_dt}
\end{proposition}
\begin{proof}
    For the DT MJLS \eqref{eq:leader-follower-reference-dt}, all eigenvalues of ${\mathcal{H}_k}$ and ${\mathcal{A}_k}$ lie strictly within the unit circle, meaning that there are steady-state solutions for $\hat{\boldsymbol{\nu}}(k)$ and $\hat{\mathbf{c}}(k)$. Since ${\hat{\boldsymbol{\nu}}}(k + 1) = {\hat{\boldsymbol{\nu}}}(k)$ at steady state, it must be true that $\hat {\boldsymbol{\nu}}_{\text{ss}} = \mathbf{0}_{nm}$ as stated in in Proposition \ref{prop_ref_tracking_cov_ss_dt}(i). Similarly, at steady state, ${\hat{\mathbf{c}}}(k + 1) = {\hat{\mathbf{c}}}(k)$ and $\boldsymbol{\pi}(k) = \boldsymbol{\pi}_{\text{ss}}$. Solving for $\hat{\mathbf{c}}(k)$ under these conditions yields the result for $\hat{\mathbf{c}}_{\text{ss}}$ stated in Proposition \ref{prop_ref_tracking_cov_ss_dt}(ii).
\end{proof}

\begin{remark}
    The results of Propositions~\ref{prop:ct-2-topology} and~\ref{prop:ct-2-topology-leader} can be easily extended to the discrete time case. Specifically, let the state transition matrix be
    \[
    P = \begin{bmatrix}
        1 - v_1 & v_1 \\ v_2 & 1- v_2
    \end{bmatrix} =  
 \begin{bmatrix}
        1 -\beta {\pi}_{\text{ss},2} & \beta {\pi}_{\text{ss},2} 
        \\
       \beta {\pi}_{\text{ss},1} & 1 - \beta{\pi}_{\text{ss},1}
    \end{bmatrix} ,
    \]
where $\beta = v_1+v_2$ and $\boldsymbol{\pi}_\text{ss} = \frac{1}{\beta} [v_2 \hspace{5pt} v_1]^\top$ is the stationary distribution. The system error is  $\text{E} = \mathbf{h}_S \widebar{\mathcal{M}}_d^+  T$, where $\widebar{\mathcal{M}}_d = R_1 - \beta R_2$ with
\begin{align*}
    R_1 &= (\mathbb{I}_{nm^2} - (\mathbb{I}_2 \otimes \mathbb{V} \otimes \mathbb{V}) \text{diag}_n (\mathbb{V} B_i\mathbb{V} \otimes \mathbb{V} B_i\mathbb{V}), \\
    R_2 &=\left(\begin{bmatrix}
        -{\pi}_{\text{ss},2} &  {\pi}_{\text{ss},2} 
        \\
    {\pi}_{\text{ss},1} & -{\pi}_{\text{ss},1}
    \end{bmatrix} \otimes \mathbb{V} \otimes \mathbb{V}\right) \text{diag}_n (\mathbb{V} B_i\mathbb{V} \otimes \mathbb{V} B_i\mathbb{V}). 
\end{align*}
Analogous steps to the proof of Proposition~\ref{prop:ct-2-topology} can be employed to derive similar expressions. \oprocend
\end{remark}

\newpage 

\section{Proof of Proposition \ref{prop_ct_cov_dyn}}\label{ap:prop_ct_cov_dyn} \label{proof_prop_ct_cov_dyn}
To prove Proposition \ref{prop_ct_cov_dyn}, we require the following lemma. 
\begin{lemma} Let event $B \subset \mathcal{J}$, where $\mathcal{J}$ is a finite sample space. Then, for a random variable $Y$ and another event $A$,
\begin{align*}
    \mathbb{E}[ Y \mathbbm{1}(A)]  = \sum_{j \in \mathcal{J}} \mathbb{E}[ Y | A \cap B = j] \Pr(A | B = j) \Pr(B = j) ,
\end{align*} 
where $j$ is an outcome of $\mathcal{J}$.
\label{lemma_prob_ind}
\end{lemma}
\begin{proof}\emph{(Lemma \ref{lemma_prob_ind})} 
    For the probability density function $f(y)$ and some event $C$, $\mathbb{E}[ Y \mathbbm{1}(C)] = \int y f(y | C) \Pr(C) dy = \mathbb{E}[Y | C] \Pr(C)$. Let $\mathbb{E}[ Y \mathbbm{1}(A)] = \sum_{j \in \mathcal{J}} \mathbb{E}[ Y \mathbbm{1}({A \cap B = j})]$.   We complete the proof by letting $C = A \cap (B = j)$.
\end{proof}

\begin{proof}\emph{(Proposition \ref{prop_ct_cov_dyn})}
Within this proof, we abbreviate $a(t) = a_t$ and $b_i(t) = b_{i,t}$. Then, for some small time increment $h$, we rewrite \eqref{eq:generic_mjls} as $\mathbf{x}_{t + h} = (\mathbb{I}_m - hZ_t) \mathbf{x}_t + F d\mathbf{W}(h)$
where $Z_t = Z_i$ if $\mathcal{G}_t = G_i \in \mathbb G$ is the network graph at time $t$. For $j \neq i$, let
\begin{align*}
    &\text{p}_{ii}(h) = \Pr (\mathcal{G}_{t+h} = G_i | \mathcal{G}_t = G_i) \approx 1 - v_i h + O(h^2)
    \\
    &\text{p}_{ji}(h) = \Pr (\mathcal{G}_{t+h} = G_i | \mathcal{G}_t = G_j) \approx q_{ji} h + O(h^2) ,
\end{align*}
where $i,j \in S$. Then, from Lemma \ref{lemma_prob_ind},
\begin{align*}
    \boldsymbol{\mu}_{t+h}^i &= \sum_{j \in S} \mathbb{E}[\mathbf{x}_{t+h} | \mathcal{G}_t = G_j] \text{p}_{ji}(h) \pi_{j,t}
    \\ &=  (\mathbb{I}_m - hZ_{i}) \boldsymbol{\mu}_{t}^i (1 - v_i h) + \sum_{j \in S \setminus i} h q_{ji} (\mathbb{I}_m - hZ_{j}) \boldsymbol{\mu}_{t}^j .
\end{align*}
We take the derivative $\boldsymbol{\dot{\mu}}_{t}^i = \lim_{h \rightarrow 0} ({\boldsymbol{{\mu}}_{t+h}^i - \boldsymbol{{\mu}}_{t}^i})/{h}$:
\begin{align*}
    \boldsymbol{\dot{\mu}}_{t}^i &= \hspace{-1pt} -(Z_{i} + v_i \mathbb{I}_m)\boldsymbol{\mu}_{t}^i  + \hspace{-1pt}\sum_{j \in S \setminus i} v_j p_{ji} \boldsymbol{\mu}_{t}^j 
    \\&= \hspace{-1pt}-Z_{i}\boldsymbol{\mu}_{t}^i  + \hspace{-1pt}\sum_{j \in S} \Gamma_{ji} \boldsymbol{\mu}_{t}^j,
\end{align*}
Vertically stacking $\boldsymbol{\dot{\mu}}_{t}^i$ for all $i \in S$ into an $mn \times 1$ vector gives the result that is stated in Proposition \ref{prop_ct_cov_dyn}(i). Similarly,
\begin{align*}
    C^i_{t+h} 
    &= \sum_{j \in S} \mathbb{E}[\mathbf{x}_{t+h}\mathbf{x}_{t+h}^\top |\mathcal{G}_{t+h} = G_i , \mathcal{G}_t = G_j] \text{p}_{ji}(h) \pi_{j,t}
    \\
    &= (C^i_t - hZ_{i} C^i_t - hC^i_t Z_{i}^\top - v_i h C^i_t) 
    \\
    &\hspace{50pt}+ \sum_{j \in S \setminus i}  C^j_t h q_{ji} + \pi_{i,t+h} h Q + O(h^2).
\end{align*}
We take the derivative $\dot{C}^i_t = \lim_{h \rightarrow 0} ({C^i_{t+h} - C^i_t})/{h}$:
\begin{align*}
\dot{C}^i_t &= -(Z_{i} C^i_t + C^i_t Z_{i}^\top + v_i C^i_t) + \hspace{-2pt}\sum_{j \in S \setminus i}  C^j_t q_{ji} + \pi_{i,t} Q
\\ &= -(Z_{i} C^i_t + C^i_t Z_{i}^\top) + \sum_{j \in S}  C^j_t \Gamma_{ji} + \pi_{i,t} Q.
\end{align*}
Then, let $\mathbf{{c}}_t^i = \text{vec}({C}^i_t)$ and vectorize the previous equation with the help of the rule $\text{vec}(AXB) = (B^\top \otimes A) \text{vec}(X)$ for matrices $A$, $X$, and $B$ \cite{horn1994topics}. As a result,
\begin{align*}
\dot{\mathbf{c}}_t^i &= -(Z_{i} \oplus Z_{i}) \textbf{{c}}_t^i  + \sum_{j \in S} \Gamma_{ji} \textbf{{c}}_t^j + \pi_{i,t} \text{vec}({Q}),
\end{align*}
which, when stacked vertically for all $i \in S$ into an $nm^2 \times 1$ vector, gives  the result that is stated in Proposition \ref{prop_ct_cov_dyn}(ii).
\end{proof}


\section{Proof of Proposition \ref{prop_dt_cov_dyn}}\label{ap:prop_dt_cov_dyn}

This proof parallels that in \ref{ap:prop_ct_cov_dyn}, which may contain information relevant to this section. First, let $W_k = W_i$ if $\mathcal{G}_k = G_i \in \mathbb G$ is the network graph at time step $k$. From  \ref{lemma_prob_ind},
\begin{align*}
    \boldsymbol{\mu}_{k}^i &= \sum_{j \in S} \mathbb{E}[\mathbf{x}_{k} | \mathcal{G}_{k-1} = G_j] {P}_{ji} \pi_{j,k-1} =  \sum_{j \in S} W_j \boldsymbol{\mu}_{k-1}^j P_{ji}.
\end{align*}
Then, vertically stacking $\boldsymbol{\mu}_{k}^i$ for all $i \in S$ into an $mn \times 1$ vector gives the result stated in \ref{prop_dt_cov_dyn}(i). Similarly, 
\begin{align*}
    C^i_{k} 
    &= \sum_{j \in S} \mathbb{E}[\mathbf{x}_{k}\mathbf{x}_{k}^\top |\mathcal{G}_{k} = G_i , \mathcal{G}_{k-1} = G_j] P_{ji} \pi_{j,k-1}
    \\
    &= \sum_{j \in S}  W_j C^j_{k-1} W_j^\top P_{ji} + \pi_{i,k} Q .
\end{align*}
Vectorizing this equation yields:
\begin{align*}
\mathbf{{c}}_k^i &= \sum_{j \in S} P_{ji} (W_j \otimes W_j) \mathbf{{c}}_{k-1}^i + \pi_{i,k} \text{vec}({Q}),
\end{align*}
which, when stacked vertically for all $i \in S$ into an $nm^2 \times 1$ vector, gives  the result that is stated in Proposition \ref{prop_dt_cov_dyn}(ii).

Similarly, the graph switching behavior of the discrete-time MSG with the graph set $\mathbb G$ is determined by a discrete-time MC (DTMC). The DTMC is described by the transition probability matrix $P \in \mathbb{R}^{n \times n}$ containing the elements: 
\begin{align*}
    {P}_{ij}(k) = \Pr (r(k) = j | r({k - 1}) = i)  ,
\end{align*}
for $i, j \in S$. Additionally,  $P$ is time-homogeneous, meaning that it is independent of $k$, and row-stochastic. As for the CTMC, the probability distribution of the DTMC at time step $k$ is given by $\boldsymbol{\pi}(k) \in \Delta_n$. In this case, $\boldsymbol{\pi}(k) = {({P}^\top)}^k \boldsymbol{\pi}(0)$.


\end{document}